\documentclass[11pt]{article}

\usepackage[letterpaper,margin=1in]{geometry}
\usepackage{setspace}
\usepackage{amsmath}
\usepackage{amsthm}
\usepackage{amssymb}
\usepackage{algorithm}
\usepackage{algpseudocodex}
\usepackage{complexity}
\usepackage{enumitem}
\usepackage[hidelinks]{hyperref}
\makeatletter
\AddToHook{env/algorithmic/before}{\def\@currentcounter{ALG@line}}
\makeatother
\usepackage{cleveref}
\crefalias{ALG@line}{line}
\usepackage{thm-restate}

\usepackage[round]{natbib}

\definecolor{CoalitionColor}{HTML}{1D42A6}

\usepackage{tikz}
\usetikzlibrary{calc}
\usetikzlibrary{decorations.pathreplacing}
\tikzstyle{agent}=[draw, shape=circle, minimum size=30pt]
\tikzstyle{smallagent}=[draw, shape=circle, minimum size=15pt]
\tikzstyle{aux}=[draw, shape=circle, minimum size=5pt]

\newcommand{\convexpath}[2]{
  [   
  create hullcoords/.code={
    \global\edef\namelist{#1}
    \foreach [count=\counter] \nodename in \namelist {
      \global\edef\numberofnodes{\counter}
      \coordinate (hullcoord\counter) at (\nodename);
    }
    \coordinate (hullcoord0) at (hullcoord\numberofnodes);
    \pgfmathtruncatemacro\lastnumber{\numberofnodes+1}
    \coordinate (hullcoord\lastnumber) at (hullcoord1);
  },
  create hullcoords
  ]
  ($(hullcoord1)!#2!-90:(hullcoord0)$)
  \foreach [
  evaluate=\currentnode as \previousnode using \currentnode-1,
  evaluate=\currentnode as \nextnode using \currentnode+1
  ] \currentnode in {1,...,\numberofnodes} {
    let \p1 = ($(hullcoord\currentnode) - (hullcoord\previousnode)$),
    \n1 = {atan2(\y1,\x1) + 90},
    \p2 = ($(hullcoord\nextnode) - (hullcoord\currentnode)$),
    \n2 = {atan2(\y2,\x2) + 90},
    \n{delta} = {Mod(\n2-\n1,360) - 360}
    in 
    {arc [start angle=\n1, delta angle=\n{delta}, radius=#2]}
    -- ($(hullcoord\nextnode)!#2!-90:(hullcoord\currentnode)$) 
  }
}

\newcommand{\agS}{N} 
\newcommand{\fracmat}{\zeta} 

\newcommand{\uf}{u} 
\newcommand{\vf}{v} 
\newcommand{\wf}{v} 

\newcommand{\SW}{\mathcal{SW}}
\newcommand{\EV}{\mathbb{E}}
\newcommand{\EVA}[1]{\EV} 
\newcommand{\prob}{\mathbb{P}}

\newcommand{\alg}{\mathit{ALG}}
\newcommand{\beste}{\mathit{ALG}^*}
\newcommand{\gentreealg}{\alg}
\newcommand{\cmax}{c^*} 
\newcommand{\ctree}{\hat c}

\DeclareMathOperator{\I}{\chi} 

\newtheorem{theorem}{Theorem}[section]
\newtheorem{corollary}[theorem]{Corollary}
\newtheorem{proposition}[theorem]{Proposition}
\newtheorem{lemma}[theorem]{Lemma}

\theoremstyle{definition}

\newcommand{\recursepara}{k}

\newcommand{\pmatched}{h} 
\newcommand{\pmatchbest}{r}

\newcommand{\pmatchedi}[1]{h} 
\newcommand{\pmatchbesti}[1]{r} 
\newcommand{\hrinput}{S} 
\newcommand{\hrinputpara}[1]{S[N\setminus\{#1\}]} 
\newcommand{\starins}{\mathcal S} 
\newcommand{\biins}{\mathcal B} 

\newcommand{\vfd}{\vf} 
\newcommand{\vb}{\left(\frac{1}{\epsilon}\right)}

\newcommand{\starmax}{t_S} 
\newcommand{\bstarmax}{t_B} 

\title{The Power of Matching for Online Fractional Hedonic Games}

\author{Martin Bullinger\\
\small University of Bristol, UK\\
\texttt{\small martin.bullinger@bristol.ac.uk}
\and
René Romen
and Alexander Schlenga\\
\small Technical University of Munich, Germany\\
\texttt{\small \{rene.romen,alexander.schlenga\}@tum.de}
}

\date{}

\allowdisplaybreaks 

\begin{document}

\maketitle

\begin{abstract}
	We study coalition formation in the framework of fractional hedonic games (FHGs).
	The objective is to maximize social welfare in an online model where agents arrive one by one and must be assigned to coalitions immediately and irrevocably.
	A recurrent theme in online coalition formation is that online matching algorithms, where coalitions are restricted to size at most~$2$, yield good competitive ratios.
	For example, computing maximal matchings achieves the optimal competitive ratio for general online FHGs.
	However, this ratio is bounded only if agents' valuations are themselves bounded.
	
	We identify optimal algorithms with constant competitive ratios in two related settings, independent of the range of agent valuations.
	First, under random agent arrival, we present an asymptotically optimal $(\frac{1}{3}-\frac 1n)$-competitive algorithm, where $n$ is the number of agents.
	This result builds on our identification of an optimal matching algorithm in a general model of online matching with edge weights and an unknown number of agents.
	In this setting, we also achieve an asymptotically optimal competitive ratio of $\frac{1}{3}-\frac 1n$.
	Second, when agents arrive in an arbitrary order but algorithms are allowed to irrevocably and entirely dissolve coalitions, we show that another matching-based algorithm achieves an optimal competitive ratio of $\frac{1}{6 + 4\sqrt{2}}$.
	\end{abstract}

\section{Introduction}

We study an online partitioning problem in which a set of agents arrives one by one and must be assigned to unique coalitions.
This setting naturally generalizes online matching by allowing coalitions of arbitrary size, and, therefore, poses fundamental combinatorial challenges for algorithm designers.
At the same time, it has many appealing applications at the intersection of economics, artificial intelligence, and the social sciences, ranging from assigning employees to projects or offices, to grouping students into project teams, patients into support groups, or organizing collaboration among robotic swarms \citep[for further applications, see][]{Ray07a, AzSa15a}.
Dependent on the context, agents may represent individuals in a society or, more broadly, firms or computer programs.
A widely studied framework for modeling coalition formation is that of \emph{hedonic games} \citep{DrGr80a, BoJa02a}, in which each agent has preferences over the coalitions they belong to.
The defining feature of hedonic games is that preferences depend solely on the members of an agent's own coalition, and not on how other agents are grouped.

However, even under this natural restriction, explicitly specifying preferences requires considering an exponentially large set of potential coalitions.
Hence, for the sake of computational tractability, much of the literature has focused on hedonic games with inherently concise preference representations.
One common approach derives an agent's preferences over coalitions from their opinion on single agents.
For example, agents may assign a subjective cardinal valuation to each other agent, which is then aggregated to determine utilities over coalitions.
This idea gives rise to, among others, the well-studied classes of additively separable (ASHGs) and fractional (FHGs) hedonic games \citep{BoJa02a, ABB+17a}.
In this work, we focus on FHGs, where an agent's utility for a coalition is the average of the valuations they assign to its members (assuming a self-valuation of~$0$).
\citet{ABB+17a} argue that this model is well-suited for analyzing network clustering and use it to represent basic economic scenarios such as the bakers-and-millers game.
Moreover, \citet{MDTB21a} apply FHG utilities in a peer-to-peer energy market, showing that the model leads to decent profits.

An important aspect of real-world coalition formation processes is that agents arrive over time.
This has motivated the study of an online model of hedonic games by \citet{FMM+21a}.
In their basic model, agents arrive one by one and have to be assigned to existing coalitions of any size immediately and irrevocably.
The objective is to achieve high social welfare, defined as the sum of agents' utilities.
Unfortunately, this is a demanding objective in FHGs: if $V_{\min}$ and $V_{\max}$ are the minimum and maximum permitted absolute value of nonzero valuations, the best possible competitive ratio is $\frac {V_{\min}}{4V_{\max}}$.

A crucial role in achieving welfare approximations, whether in an offline or online setting of coalition formation, has been to employ matchings.\footnote{A notable exception are online FHGs with nonnegative weights, for which the optimal algorithm forms coalitions of unbounded size \citep{FMM+21a}.} 
For instance, the aforementioned competitive ratio is attained by forming maximal matchings, which is even the best deterministic approach for unweighted games \citep{FMM+21a}.
Moreover, the best known polynomial-time approximation algorithm for social welfare in offline FHGs, achieving a $2$-approximation, is to form a maximum weight matching \citep{FKMZ21a}.
Similarly, in the related model of ASHGs, maximum weight matchings achieve an $n$-approximation of social welfare, where $n$ is the number of agents.
At the same time, an $n^{1-\epsilon}$-approximation is \NP-hard to compute for any $\epsilon > 0$ \citep{FKV22a}, even if $V_{\min} = V_{\max} = 1$ \citep{BCS25a}.
Our work extends this intuition by considering two more sophisticated models of online FHGs, where we show that online matching algorithms achieve a constant optimal competitive ratio.

In the first model, the random arrival setting, 
agents arrive in a uniformly random instead of an adversarial order, as in the secretary problem \citep{Ferg89a}.
Hence, an algorithm has to compete well against an adversary who only fixes the game but not the precise arrival order.
Notably, this avoids the worst-case example by \citet{FMM+21a}, which crucially relies on specifying valuations based on the previous decisions of algorithms.
We achieve a $(\frac 13- \frac 1n)$-competitive algorithm, while no algorithm can be better than $\frac 13$-competitive.
The algorithm we present is based on an online matching algorithm for a general setting of weighted matching where all agents arrive online but the number of agents is unknown.
In the matching domain, we obtain the same asymptotically optimal competitive ratio. 
The tightness of both competitive radii follows a unified approach:
We establish an upper bound for the competitive ratio of any online matching algorithms on the tree domain, a specific domain of instances where positive valuations form trees.
We then show that upper bounds on this domain directly transfer to the coalition formation setting.

In the second model, the free dissolution setting, the algorithm has to perform well under any arrival order, but it gets the additional power to dissolve coalitions.
This setting is inspired by free edge dissolution in the matching domain \citep{FKM+09a} and has been introduced to coalition formation by \citet{BuRo25b}.
We show that another matching algorithm achieves the optimal competitive ratio of $\frac{1}{6+4\sqrt{2}}$, which is a factor $\frac 12$ worse than the best online algorithm in the corresponding matching domain.

\section{Related Work}

The hedonic formation of coalitions traces back to \citet{DrGr80a}, while hedonic games in the form studied today have been conceptualized by \citet{BoJa02a}.
The latter paper introduces the class of ASHGs, in which utilities for coalitions are obtained through a sum-based aggregation of individual valuations. 
Fractional hedonic games were introduced later by \citet{ABB+17a}.
An overview of hedonic games can be found in the book chapters by \citet{AzSa15a} and \citet{BER24a}.

Several authors studied various notions of stability in FHGs \citep{BBS14a,BFF+15a,BFFMM18a,KKP16a,ABB+17a,BrBu20a}, where the goal is outcomes in which agents cannot perform beneficial deviations to a more preferred coalition.
\citet{AGG+15b} consider welfare maximization.
In addition to examining algorithms for (utilitarian) social welfare, they consider the maximization of egalitarian and Nash welfare.
They prove \NP-hardness of finding optimal partitions for the different
objectives and give polynomial-time approximation algorithms for utilitarian and egalitarian welfare.
Matching algorithms are shown to yield good approximation ratios.
In particular, \citet{AGG+15b} show that a maximum weight matching (MWM) is a
$\frac{1}{4}$-approximation of social welfare in 
FHGs.
This analysis was later improved and made tight by \citet{FKMZ21a} who prove that MWMs yield precisely a
$\frac{1}{2}$-approximation.

An online model for hedonic games was first studied by \citet{FMM+21a}, who consider FHGs and ASHGs.\footnote{
While being inspired by models of hedonic games, \citet{FMM+21a} develop their model as a ``coalition structure generation problem'' and, therefore, adopt a purely graph-theoretic instead of a game-theoretic perspective.
}
They investigate the model where agents arrive in an adversarial order
and give lower and upper bounds for deterministic algorithms on
the achievable competitive ratio for maximizing social welfare.
Except for simple FHGs, their results are rather discouraging because the competitiveness crucially depends on the range of valuations.
For ASHGs, \citet{BuRo25b} consider the random arrival and the free dissolution models and show that these dependencies vanish.
We achieve similar results for FHGs.
Furthermore, going beyond welfare maximization, \citet{BuRo25a}
study stability and Pareto optimality
for online ASHGs with adversarial agent arrival.

There is a vast body of literature on online matching.
A recent survey is given by \citet{HTW24a}.
Here, we only discuss the works that are closest to our setting.
For unweighted graphs, \citet{GKM+19a} give the online algorithm with the currently best known competitive ratio for maximum cardinality matchings with adversarial vertex arrival.
\citet{KRTV13a} study MWMs with random vertex arrival on one side of bipartite graphs and show
that the upper bound of~$\frac{1}{e}$, which stems from the fact that the scenario generalizes the secretary problem,
can be matched by an algorithm.
\citet{EFGT22a} present an algorithm for approximating an MWM in general weighted graphs with random vertex arrival
where the total number of vertices to arrive is known in advance.
They also show the asymptotic tightness of that algorithm's competitive ratio by considering a family of graphs where all edge weights
differ by a large factor, so there is only one valuable edge for a matching.
Our optimal matching algorithm considers their setting with an unknown number of agents.
Finally, \citet{FKM+09a} introduce free disposal of edges in online matching. 
They determine the optimal competitive ratio of $1-\frac 1e$ for the bipartite setting when one side of the agents is present offline.  
\citet{BuRo25b} extend free disposal to the notion of free dissolution in coalition formation that we study.

\section{Preliminaries and Model}

We begin by introducing some 
notation.
For $i\in\mathbb{N}$, we denote $[i] := \{1,\dots,i\}$.
Next, for a graph $G=(V, E)$ and a set of vertices $S\subseteq V$, let $G[S]$ denote the subgraph of $G$ induced by~$S$.
Finally, we denote the indicator function by $\I(\cdot)$.
It takes a Boolean argument as an input and returns~$1$ if it is true and~$0$, otherwise.

\subsection{Hedonic Games}
Let $N$ be a finite set of \emph{agents}. 
A nonempty subset $C\subseteq N$ is called a \emph{coalition}.
The set of coalitions containing agent~$i\in N$ is denoted by 
$\mathcal N_i:=\{C\subseteq N\mid i\in C\}$.
A set $\pi$ of disjoint coalitions containing all members of $N$ is a \emph{partition} of $N$.
For agent $i\in N$ and partition $\pi$, let $\pi(i)$ denote the unique coalition in~$\pi$ that~$i$ belongs to.
Moreover, for a subset of agents $N'\subseteq N$, we define $\pi[N']$ as the \emph{partition restricted to $N'$} as $\pi[N'] := \{C\cap N'\mid C\in \pi, C\cap N'\neq \emptyset\}$. 

A (cardinal) \emph{hedonic game} is a pair $G = (N,\uf)$ where $N$ is the set of agents and $\uf = {(\uf_i)}_{i\in N}$ is a tuple of \emph{utility functions} $u_i\colon \mathcal N_i \to \mathbb Q$.
Agents seek to maximize utility and prefer partitions in which their coalition achieves a higher utility.
Hence, we define the utility of a partition~$\pi$ for agent~$i$ as $\uf_i(\pi) := \uf_i(\pi(i))$.
We denote by $n(G):= |N|$ the number of agents and write $n$ if $G$ is clear from the context.

Following \citet{ABB+17a}, a \emph{fractional hedonic game} (FHG) is a hedonic game $(N,\uf)$, where for each agent $i\in N$ there exists a \emph{valuation function} $\vf_i\colon N\setminus\{i\}\to \mathbb Q$ such that for all $C\in \mathcal N_i$ it holds that $\uf_i(C) = \sum_{j\in C\setminus \{i\}}\frac{\vf_i(j)}{|C|}$.
Note that this implies that the utility for a singleton coalition is~$0$.
Since the valuation functions contain all information for computing utilities, we also represent an FHG as the pair $(N,\vf)$, where $\vf = {(\vf_i)}_{i\in N}$ is the tuple of valuation functions.
Additionally, an FHG can be succinctly represented as a complete directed weighted graph where the weights of directed edges induce the valuation functions.

An FHG $(N,\vf)$ is said to be \emph{symmetric} if for every pair of distinct agents $i,j\in N$, it holds that $\vf_i(j) = \vf_j(i)$.
We write $\vf(i,j)$ for the symmetric valuation between~$i$ and~$j$.
A complete undirected weighted graph can represent a symmetric FHG\@.
For simplicity, we also denote this graph by $(N,\vf)$.
Moreover, an FHG is said to be \emph{simple} if for every pair of distinct agents $i,j\in N$, it holds that $\vf_i(j)\in \{0,1\}$. 
Simple FHGs can be represented by
directed unweighted graphs (where edges represent valuations of~$1$).
Finally, a symmetric FHG is said to belong to the
\emph{tree domain} if every connected component of the edges with positive weight in the
associated undirected graph forms a tree,
and every other edge has a negative weight smaller than
the negative sum of all positive edge weights.

We measure the desirability of a partition in terms of social welfare.
Given an FHG $G = (N,\vf)$, we define the \emph{social welfare} of a coalition $C\subseteq N$ in $G$ as $\SW_G(C) := \sum_{i\in C}\uf_i(C)$ and of a
partition $\pi$ in $G$ as $\SW_G(\pi) := \sum_{i\in N}\uf_i(\pi) = \sum_{C\in \pi}\SW_G(C)$.
We omit the subscript $G$ if the game is clear from the context.
Also, we denote by $\pi^*(G)$ a partition that maximizes social welfare in $G$.
Note that we can replace both $\vf_i(j)$ and $\vf_j(i)$ by $\frac 12 (\vf_i(j)+\vf_j(i))$ for all $i,j\in N$, which results in a symmetric FHG in which the social welfare of every partition remains the same \citep{Bull19a}.
Hence, it suffices to consider symmetric FHGs instead of the full domain of FHGs.
However, note that this technique cannot be applied to simple FHGs (or other restricted classes of FHGs) as the symmetrization may result in nonsimple FHGs.
Given $c\le 1$, a partition $\pi$ is called a \emph{$c$-approximation} to social welfare in $G$ if $\SW(\pi)\ge c\cdot \SW(\pi^*(G))$.

\subsection{Matching}

A \emph{matching} is a partition in which all coalitions have size at most~$2$.\footnote{In contrast to the standard definition of a matching, we assume that unmatched agents are part of a matching in the form of singleton coalitions.}
For a matching $\pi$, we denote by $\vf(\pi) = \sum_{e = \{i,j\}\in \pi, |e| = 2} \vf(e)$ its \emph{weight}. 
We have that $\SW(\pi) = \vf(\pi)$ because, for each matched pair, both agents contribute $\frac 12$ of the edge weight to the social welfare.
Hence, maximizing social welfare among matchings is precisely the \emph{maximum weight matching} (MWM) problem.

Assume that we are given a complete unweighted graph $G = (N,\vf)$.\footnote{The below definitions also work for incomplete graphs which we complete by adding edges with weight~$0$.}
A \emph{fractional matching} is a function $\fracmat\colon V \times V \to [0,1]$ such that, for all $i,j\in N$, we have that $\fracmat(i,j) = \fracmat(j,i)$ and $\sum_{j\in N} \fracmat(i,j) = 1$.
The value $\fracmat(i,j)$ is interpreted as the probability of matching $i$ with $j$, where we think of $\fracmat(i,i)$ as the probability that~$i$ remains unmatched.
This interpretation makes sense because of the second constraint and we can extract the probability distribution of matching to neighbors as follows.
For $i\in N$, we define $\fracmat(i)\colon V \to [0,1]$ by $\fracmat(i)(j) = \fracmat(i,j)$.

Given a fraction matching $\fracmat$, its weight is defined by $\vf(\fracmat) :=\sum_{i,j\in N, i\neq j} \fracmat(i,j)\vf(i,j)$.
We use fractional matchings as a computational tool as a subroutine for one of our algorithms.
For this, note that the problem of computing a \emph{maximum weight fractional matching} (MWFM) is well known to be solvable in polynomial time \citep[see, e.g.,][]{Schr03a}. 

\subsection{Online Models and Competitive Analysis}

We assume an online model of FHGs where agents arrive one by one and have to be assigned to new or existing coalitions.
For an agent set $N$, define $\Sigma(N) := \{\sigma\colon [|N|]\to N \textnormal{ bijective}\}$.
This is interpreted as the set of all \emph{arrival orders}.

An instance $(G,\sigma)$ of an \emph{online FHG} consists of an FHG $G = (N,\vf)$ and an arrival order $\sigma\in \Sigma(N)$.
An online coalition formation algorithm $\alg$ produces on input $(G,\sigma)$ a sequence ${\alg(G,\sigma)}_1$, \dots, ${\alg(G,\sigma)}_{n(G)}$ of partitions, where for each $i\in [n(G)]$, ${\alg(G,\sigma)}_i$ is a partition of $\{\sigma(1),\dots, \sigma(i)\}$.
Hence, the partial partitions have to contain precisely the agents that have arrived so far.
Moreover, we require that for all input tuples $(G,\sigma)$ and $(H,\tau)$ and $k\in \mathbb N$ with $k\le \min \{n(G),n(H)\}$ it holds that
${\alg(G,\sigma)}_k = {\alg(H,\tau)}_k$ whenever $\vf_{\sigma(i)}(\sigma(j)) = \vf_{\tau(i)}(\tau(j))$ for all $i,j\in [k]$.\footnote{We later consider randomized algorithms, for which the produced random partition has to be identical.}
This condition says that the algorithmic decision to form the $k$th partition can only depend on the information the algorithm has obtained until the $k$th agent arrives.
In particular, it cannot depend on the knowledge about agents arriving in the future.
Furthermore, this condition implies that decisions must be identical if all valuations are identical up to a certain agent's arrival.
The output of the algorithm is the partition produced when the final agent is added; we denote $\alg(G,\sigma):= {\alg(G,\sigma)}_{n(G)}$.

In addition, an algorithm's decisions are assumed to be irrevocable, i.e., agents can only be added to an existing or a completely new coalition, while not changing the existing coalitions.
Formally, this means that for all instances $(G,\sigma)$ and $2\le k\le n(G)$, we require that ${\alg(G,\sigma)}_k[\{\sigma(i)\mid 1\le i\le k-1\}] = {\alg(G,\sigma)}_{k-1}$, i.e., the $(k-1)$st partition is the $k$th partition restricted to the first $k-1$ agents.
An algorithm may, however, have the additional power to dissolve a partition before adding a new agent.
In this case, we say that the algorithm operates under \emph{free dissolution} and additionally allow that ${\alg(G,\sigma)}_k[\{\sigma(i)\mid 1\le i\le k-1\}]$ is of the form $({\alg(G,\sigma)}_{k-1}\setminus C) \cup \{\{i\}\mid i \in C\}$ for some $C\in {\alg(G,\sigma)}_{k-1}$.

The objective is to achieve high welfare. 
We say that $\alg$ is \emph{$c$-competitive}\footnote{We use the convention that $\frac 00 = 1$ and $\frac x0 = 0$ for any $x\in \mathbb Q$ with $x < 0$.
} if
\[
\inf_G\min_{\sigma\in \Sigma(N)}\frac{\SW(\alg(G,\sigma))}{\SW(\pi^*(G))}\ge c\text.
\]

Equivalently, this means that for all instances $(G,\sigma)$, $\alg$ produces a $c$-approximation of social welfare.
Hence, we benchmark algorithms against a worst-case adversary that can both fix an instance, i.e., the number of agents and their mutual valuations, as well as an exact arrival order.

In addition, we consider a model where the agents arrive in a \emph{uniformly random} arrival order. 
The objective is then to achieve high welfare in expectation.
We denote by $\alg(G)$ the random partition produced with respect to a uniformly random arrival order.
An algorithm $\alg$ is said to be \emph{$c$-competitive under random arrival} if
\[
\inf_G\frac{\EV_{\sigma\sim\Sigma(N)}\left[\SW(\alg(G))\right]}{\SW(\pi^*(G))}\ge c\text.
\]
Hence, in this model, an algorithm is benchmarked against an adversary that can design a worst-case instance, but has no control over the exact arrival order of the agents.
We omit the subscript from the notation of expected values when it is clear from the context.
In both models, the \emph{competitive ratio} $c_\alg$ of $\alg$ is the supremum $c$ such that $\alg$ is $c$-competitive.
Note that the competitive ratio is always at most~$1$.

We also consider randomized algorithms, which can use randomization to decide which coalition an agent should be added to.
In this case, the competitive ratio is measured with respect to the expected social welfare of the random partition constructed by the randomized algorithm.

The competitive ratio is also defined for subclasses of FHGs, such as simple and symmetric FHGs, where the infimum is only taken over games from that subclass.
Finally, the competitive ratio is also defined for online matching algorithms, for which the weight of the matching produced by an algorithm is compared with the weight of an MWM\@.

\section{Connections between Matching and FHGs}\label{sec:connections}

A first significant connection between MWMs and welfare maximization in FHGs is that the former yields a $\frac 12$-approximation for the latter.
In \Cref{app:prelimresults}, we provide an instructive alternative proof of this theorem, which was first shown by \citet{FKMZ21a}.
Our argument establishes a connection between a uniform fractional matching, where each edge is included with probability~$\frac 1n$ and the social welfare of a coalition.
The result then follows as the weight of a uniform fractional matching is a lower bound on the weight of an MWM\@.

\begin{restatable}{theorem}{matchingguarantee}[\citet{FKMZ21a}]\label{thm:matchingguarantee}
	Every MWM is a $\frac 12$-approximation of social welfare in symmetric FHGs. 
\end{restatable}

This implies the same guarantee for online algorithms: $c$-competitive online matching algorithms are $\frac c2$-competitive for online FHGs.
We can use this insight to make an interesting observation:  it is known that no deterministic online algorithm can achieve a competitive ratio of better than $\frac 14$ for simple symmetric FHGs \citep{FMM+21a}.
However, there exists a \emph{randomized} online matching algorithm for MWM on unweighted graphs (i.e., maximum cardinality matching) that beats a competitive ratio of $\frac 12$ \citep{GKM+19a}, i.e., achieves a competitive ratio of $\frac 12 + 2\epsilon^*$ for some constant $\epsilon^* > 0$.
We can apply \Cref{thm:matchingguarantee} to conclude that randomization can be utilized to beat the best deterministic algorithm in this case.

\begin{corollary}
	There exists $\epsilon^* > 0$ and a randomized online coalition formation algorithm for simple and symmetric FHGs with competitive ratio $\frac 14 + \epsilon^*$.
\end{corollary}

In contrast to \Cref{thm:matchingguarantee}, negative results for MWM, i.e., impossibilities of achieving a certain competitive ratio, do not transfer to FHGs.
They only imply that it is impossible to create a matching of a certain quality.
This does not rule out that an online algorithm can create a partition with larger coalitions that achieve more social welfare.
However, we now show that negative results are inherited on domains where positive valuations form a forest
(while other valuations are sufficiently negative).
We provide the proof in \Cref{app:prelimresults}.

\begin{restatable}{proposition}{treenegative}\label{prop:treenegative}
	Let $c\le 1$ and assume that no $c$-competitive (randomized) algorithm exists for online MWM on the tree domain.
	Then, no $c$-competitive (randomized) online coalition formation algorithm exists for symmetric FHGs. 
\end{restatable}

Interestingly, negative results for MWM are usually essentially achieved on the tree domain \citep{Vara11a, KRTV13a, BST18a, BuRo25b},\footnote{These constructions usually contain $0$-weights, which can be replaced with large negative weights.
In the setting of \citet{KRTV13a}, the secretary problem marks the hardest instances, so they are in the tree domain.} which makes the previous theorem very powerful.
However, even if we have a tight result for MWM where the lower bound is achieved on the tree domain, \Cref{thm:matchingguarantee,prop:treenegative} leave a gap of a factor of~$2$.
As we will see, closing this gap can take significant effort and the exact optimal competitive ratio can be at either extreme.

\section{Fractional Hedonic Games under Random Arrival}\label{sec:randomarrival}

While under adversarial arrival, forming maximal matchings constitutes the best deterministic online algorithm \citep{FMM+21a}, this and other greedy approaches prove suboptimal under random arrival.
\Citet{BuRo25b} describe a family of worst-case instances for greedy algorithms in the random arrival ASHG model,
and show that they impose an upper bound of $\mathcal{O}\left(\frac1{n^2}\right)$ on the competitive ratio of such algorithms.
To be more precise, the bound holds for every (randomized) algorithm which always
assigns a newly arrived agent to an existing coalition whenever this increases the current social welfare.
Since those instances are in the tree domain, algorithms will, without loss of generality, only form matchings.
Matchings yield essentially the same welfare in the ASHG and FHG model, only differing by a constant
factor of $2$ (in ASHGs we do not divide by the coalition size of $2$).
Hence, the above bound on the competitive ratio of greedy algorithms holds for our setting, too.

Still, based on \Cref{sec:connections}, a reasonable strategy to obtain good online algorithms for FHGs is to consider better online matching algorithms.
For the matching setting under random arrival, \citet{EFGT22a} 
provide an algorithm that achieves a competitive ratio of $\frac{5}{12} - \mathcal{O}(\frac{1}{n})$ if the algorithm has access to the number of agents $n$.
Importantly, knowledge of $n$ is crucial for achieving this competitive ratio. 
In the first phase of the algorithm, a subset of $k$ agents is not matched at all, and the optimal competitive ratio is achieved for $k:=\left\lfloor \frac n2\right\rfloor$.
However, one can also apply their algorithm by setting $k$ to a fixed constant. 
By setting $k = 3$, for instance, one obtains an online matching algorithm that is $\frac{1}{3} - \mathcal{O}(\frac{1}{n})$-competitive.
We get the following theorem.

\begin{restatable}{theorem}{onlineMWM}\label{thm:online_mwm_unknown_n}
	There exists a randomized online matching algorithm with a competitive ratio under random arrival of at least
	$\frac{1}{3}-\mathcal{O}(\frac{1}{n})$.
\end{restatable}

By applying \Cref{thm:matchingguarantee}, we can interpret this algorithm as a coalition formation algorithm, which implies a competitive ratio of $\frac{1}{6}-\mathcal{O}(\frac{1}{n})$ in the coalition formation domain.
Its competitive ratio in this domain is upper bounded by $\frac5{18}$ (see \Cref{app:518}).
However, we will now prove that we can even achieve a competitive ratio of $\frac{1}{3}-\mathcal{O}(\frac{1}{n})$ in the coalition formation domain.

\begin{algorithm}[tb]
	\caption{Online matching and coalition formation under random arrival}\label{algorithm:onlinemwm}
	\textbf{Input:} Complete and weighted undirected graph $G=(\agS,\wf)$, arrival order $a_1,\dots, a_n$ of the vertices, 
	parameter $k$
			\\
	\textbf{Output:} Matching $\mu$ of $G$
	\begin{algorithmic}[1]
		\State{} $A:=\agS$, $\mu:=\emptyset$
		\Comment{Initialize set of available vertices $A$, returned matching $\mu$}
		\For{$t=1$ to $k$}
		\State Add $\{a_t\}$ to $\mu$.
		\EndFor
		\For{$t=k+1$ to $n$}
		\State{} Let $\agS_t:=\{a_1,\dots,a_t\}$.
		\Comment{$\agS_t$ is the set of vertices arrived up to time $t$}
		\State{} Let $\fracmat_t$ be the MWFM in $G[\agS_t]$.
		In the case of multiple MWFMs in $G[\agS_t]$, we (randomly) choose one of them
		\emph{independent} of $A$ and the arrival order up to time $t$.\label{ln:independenceMWFM}
		
		\Comment{Recall that $\fracmat(a_t)$ is the probability distribution of matching $a_t$ to agents in $N_t$, where $a_t$ remains unmatched with probability $\fracmat(a_t,a_t)$}
		\State{} Choose a random agent $p_t\in N_t$ according to distribution $\fracmat_t(a_t)$. 
		\State Let $e_t:=\{a_t,p_t\}$. \Comment{I.e., $e_t=\{a_t\}$ if $p_t = a_t$}
		\If{$p_t = a_t$}
		\State{} With probability $\frac 13 + \frac{2(t-4)!k!}{3(t-1)!(k-3)!}$, remove $a_t$ from $A$.\label{line:X}
		\State{} Add $\{a_t\}$ to $\mu$. 
		\Else{}
		\If{$p_t\in A$}
		\State{} Remove $a_t$ and $p_t$ from $A$.
		\State{} Remove $\{p_t\}$ from $\mu$ and add $e_t$ to $\mu$.
		\Comment{Add the chosen edge to the matching}
		\Else{}
		\State{} Add $\{a_t\}$ to $\mu$.
		\EndIf{}
		\EndIf{}
		\For{$x\in (\agS_t\cap A)\setminus e_t$}
		\State{} Remove $x$ from $A$ with probability $\frac{\fracmat_t(x,x)}{t-2+\fracmat_t(x,x)}$.\label{line:remove}
		\EndFor{}
		\EndFor{}
		\State{} \Return{} Matching $\mu$ 
	\end{algorithmic}
\end{algorithm}

Consider \Cref{algorithm:onlinemwm}.
This algorithm is once again a matching algorithm but we will show that it achieves the desired guarantee for FHGs.
It takes as input an FHG (in its representation as a weighted graph) together with an arrival order and returns a partition that is obtained by assigning agents to coalitions of size at most~$2$ one by one according to their arrival.

Similar to the algorithm by \citet{EFGT22a}, \Cref{algorithm:onlinemwm} is parameterized by a natural number $k\ge3$, which indicates a ``waiting phase'' of the algorithm, i.e., the first $k$ agents are placed in singleton coalitions.
The main step of the algorithm follows a powerful idea applied in various online matching algorithms \citep{KRTV13a,EFGT22a}:
Whenever an agent arrives, we compute a maximum weight matching.
We then observe the edge to which the newly arrived agent is incident. 
If the other agent of this edge is still available, i.e., presently in a singleton coalition, we form a new coalition of size~$2$ and mark both agents as unavailable.
However, there is a crucial difference: 
Unlike \citet{KRTV13a} and \citet{EFGT22a}, we do not extract the possible coalition from a maximum weight \emph{integral} matching, but from a maximum weight \emph{fractional} matching.
In fact, using a maximum weight integral matching does not yield the desired guarantee for FHGs.
We provide an example that highlights the difficulties in \Cref{app:518}.
In contrast, fractional matchings lead to a relaxation of the integral matching polytope, and can, therefore, achieve a higher maximum weight.
Now, extracting an edge from the MWFM leads to a higher expected weight.
While this does not improve the asymptotic worst-case performance as a matching algorithm, it improves the performance as a coalition formation algorithm.

Another crucial feature of the algorithm is that the computed MWFM in \cref{ln:independenceMWFM} of \Cref{algorithm:onlinemwm} has to be chosen independent of the set of still available agents and the specific arrival order that has led to the currently present agents.
We can do this by running any algorithm for MWFM that we run for a uniformly random renaming of the agents.\footnote{\citet{KRTV13a} and \citet{EFGT22a} achieve a similar property for MWMs by randomly perturbing all weights with a small number $\epsilon$, which guarantees uniqueness of the MWM\@.
However, we prefer a method without perturbation because choosing a suitable perturbation parameter without knowledge of the number of arriving agents would cause further complications.}

We are ready to present our proof.
To make it accessible more quickly, we defer the proofs of intermediary lemmas to \Cref{app:randomarrival:optalg}.

\begin{restatable}{theorem}{CofoRandArr}\label{thm:CofoRandArr}
	There exists a randomized online coalition formation algorithm with a competitive ratio under random arrival of at least
	$\frac{1}{3}-\frac{1}{n}$.
\end{restatable}

\begin{proof}
	Throughout the proof we refer to \Cref{algorithm:onlinemwm} as $\alg$.
	Whenever we run it for a given FHG, we assume a random arrival order.
	
	Let $A_t$ be the set of available agents (i.e., the current set $A$) after $a_t$ has arrived and been processed.
	The first key step is to derive a formula for the probability that an agent is unavailable that only depends on the current iteration $t$ as well as the number $k$ of agents in the waiting phase.
	In particular, this probability is the same for every agent.
	
	\begin{restatable}{lemma}{twothirdsmatched}\label{lem:twothirdsmatched}
		Assume that we run $\alg$ for any FHG $(N,\vf)$ and a random arrival order. 
		For every $k\ge3$, $t\in \{k,\dots, n\}$, every possible realization $\tilde{\agS}$ of $\agS_t$ (i.e., $\tilde{\agS}\subseteq \agS$, $\lvert\tilde{\agS}\rvert=t$), and every agent $a\in\tilde{\agS}$, it holds that
		\begin{equation*}
			\prob\left(a\in \agS_t\setminus A_t\mid \agS_t=\tilde{\agS}\land a\in\tilde{\agS}\right)=\frac23\left(1-\frac{(t-3)!\cdot k!}{t!\cdot(k-3)!}\right)\text.
		\end{equation*}
	\end{restatable}
	
The next key step is to derive a monotonicity property when running $\alg$. 
Informally, we show that the performance cannot worsen if we increase any valuation.

\begin{restatable}{lemma}{monotonic}\label{lem:monotonic}
	Let $G=(\agS,\vf)$, and $G'=(\agS,\vf')$ be two symmetric FHGs with $\vf(i,j)\ge \vf'(i,j)$ for all $i,j\in N$.
	Then, 
	\begin{equation*}
		\EVA{N}\left[\SW_G\left(\alg(G)\right)\right]
		\ge\EVA{N}\left[\SW_{G'}\left(\alg(G')\right)\right]\text.
	\end{equation*}
\end{restatable}

	The high-level idea for the proof of \Cref{lem:monotonic} is as follows.
	The produced social welfare of $\alg$ only depends on the probability that agents are matched and the expected weight of the maximum fractional matching in each step.
	Now, by \Cref{lem:twothirdsmatched}, increasing a single valuation does not change the probability of agents being matched, while the expected weight of the maximum fractional matchings can only become larger.

	The third step is to prove the desired performance of $\alg$ on the specific domain of FHGs where the valuations between agents are negative whenever they are between agents of different coalitions of some given optimal partition.

	\begin{restatable}{lemma}{IslandInstances}\label{lem:island_instances}
		Assume that we run $\alg$ for $k=3$.
		Consider an FHG $G$ together with an optimal (offline) partition $\pi^*(G)$.
		Assume that for all $C_1,C_2\in \pi^*(G)$ with $C_1\neq C_2$, $x\in C_1$, and $y\in C_2$, we have that $\vfd(x,y)<0$.
		Then, it holds that
		\begin{equation*}
			\frac{\EVA{N}\left[\SW\left(\alg(G)\right)\right]}{\SW\left(\pi^*(G)\right)}
			\ge \frac13-\frac1n\text.
		\end{equation*}
	\end{restatable}
	
	We can combine the insights of \Cref{lem:monotonic,lem:island_instances} to prove our theorem.
	Consider any FHG $G = (N,\vf)$ together with an optimal partition $\pi^*(G)$.
	We obtain a related instance $G' = (N,\vf')$ by making all valuations between different coalitions of $\pi^*(G)$ negative, i.e.,
	
	\begin{equation*}
		\vf'(x,y) = \begin{cases}
			\vf(x,y) & \text{ if } \exists C\in \pi^*(G) \text{ with } x,y \in C \text{ or } \vf(x,y)<0\text,\\
			-1 & \text{ else.}
		\end{cases}
	\end{equation*}
	
	Note that $G'$ differs from $G$ only by decreasing valuations.
	Hence, \Cref{lem:monotonic} implies that $\EVA{N}\left[\SW_G\left(\alg(G)\right)\right]
		\ge\EVA{N}\left[\SW_{G'}\left(\alg(G')\right)\right]$.
	Moreover, for any partition $\pi$ of $N$, it holds that 
	$$\SW_{G'}\left(\pi^*(G)\right) = \SW_{G}\left(\pi^*(G)\right) \ge \SW_{G}\left(\pi\right) \ge \SW_{G'}\left(\pi\right)\text.$$
		
	In the equality, we use that valuations in $G$ and $G'$ are identical within coalitions of $\pi^*(G)$.
	Hence, $\pi^*(G)$ is also maximizes welfare with respect to $G'$.
	Now note that $G'$ satisfies the conditions of  \Cref{lem:island_instances}.
	Thus, we have that
	\begin{equation*}
		\frac{\EVA{N}\left[\SW_{G'}\left(\alg(G')\right)\right]}{\SW_{G'}\left(\pi^*(G)\right)}
		\ge \frac13-\frac1n\text.
	\end{equation*}
	
	We conclude that
	\begin{align*}
		\frac{\EVA{N}\left[\SW_G(\alg(G))\right]}{\SW_{G}(\pi^*(G))} 
		\ge \frac{\EVA{N}\left[\SW_{G'}\left(\alg(G')\right)\right]}{\SW_{G}\left(\pi^*(G)\right)}
		= \frac{\EVA{N}\left[\SW_{G'}\left(\alg(G')\right)\right]}{\SW_{G'}\left(\pi^*(G)\right)}
		\ge \frac13-\frac1n\text.
	\end{align*}
	Since $G$ was an arbitrary FHG, we conclude that $\alg$ has the desired competitive ratio.
\end{proof}

We remark that \Cref{algorithm:onlinemwm} is a matching algorithm because all returned coalitions are of size at most~$2$.
Hence, \Cref{thm:CofoRandArr} even applies a competitive ratio of $\frac 13-\frac 1n$ in the matching domain.
Thus, it presents an alternative algorithm for achieving \Cref{thm:online_mwm_unknown_n}.

As we show next, an asymptotic competitive ratio of $\frac 13$ as achieved in \Cref{thm:online_mwm_unknown_n,thm:CofoRandArr} is optimal in both the matching and coalition formation domain.
In particular, this also means that the competitive ratio of $\frac 5{12}$ in the matching domain as achieved by \citet{EFGT22a} when the number of agents is known is off limits.

Since our proof is rather long, we discuss the main idea here and defer a formal proof to \Cref{app:randomarrival:limitmain}.
In essence, our construction relies on a careful interplay of two sets of instances whose positive edges form stars and bi-stars, i.e., a union of two stars whose centers are connected by an additional edge. 
This forces an algorithm to an undesired trade-off:
The optimal matching in a star is to match the edge with the largest weight.
In our bi-stars, the largest weight is the edge connecting the two centers, so the optimal matching contains exactly this edge.
Now, by design of our instances, until both centers have arrived, an algorithm cannot distinguish whether its input is a star or a bi-star.
The key step is to show that a competitive ratio of $\frac 13$ on a star can only be achieved if matching an edge with roughly a probability of at least $\frac 23$.
However, as we will show, this means that when we are in a bi-star 
then the algorithm can only succeed with a probability of about $\frac 13$.
This leads to a bound of the competitive ratio by $\frac 13$.

\begin{restatable}{theorem}{MWMrandNegative}\label{thm:negativeresult_MWMrand}
	No randomized online matching algorithm has a competitive ratio under random arrival of more than
	$\frac{1}{3}$ on the tree domain.
\end{restatable}

Combining \Cref{thm:negativeresult_MWMrand} with \Cref{prop:treenegative}, we conclude that no online coalition formation algorithm has a competitive ratio under random arrival of more than $\frac{1}{3}$.

\begin{corollary}\label{cor:RAcofoBound}
	No randomized online coalition formation algorithm has a competitive ratio under random arrival of more than
	$\frac{1}{3}$.
\end{corollary}

\section{Fractional Hedonic Games under Coalition Dissolution}\label{sec:freediss}

We first consider the setting where algorithms should perform well regardless of a fixed arrival order but where algorithms can dissolve coalitions.
In this setting, there exists a deterministic online matching algorithm achieving a competitive ratio of $\frac{1}{3+2\sqrt{2}}$ \citep{McGr05a, BuRo25b}.\footnote{
	\citet{McGr05a} achieves this competitive ratio in the much related edge arrival model. 
	\citet{BuRo25b} showed that it is preserved in the vertex arrival model. 
}
We can apply \Cref{thm:matchingguarantee} to obtain an algorithmic guarantee for FHGs.

\begin{theorem}\label{thm:dissolutionalg}
	There exists a deterministic online coalition formation algorithm operating under free dissolution with a competitive ratio of at least $\frac{1}{6+4\sqrt{2}}$.
\end{theorem}

The algorithm mentioned above is optimal for the matching domain in the tree domain \citep{Vara11a}.
By \Cref{prop:treenegative}, no deterministic online algorithm is better than $\frac{1}{3+2\sqrt{2}}$-competitive.
We can, however, improve upon this result by proving a bound matching \Cref{thm:dissolutionalg}.

We illustrate here the main ideas for its proof and defer the full proof to \Cref{app:dissolutionbound}.
The proof technique is similar to the proof by \citet{Vara11a} in the matching domain.
However, we construct an enhanced version of the adversarial instance, where the partitions produced by an algorithm continue to be matchings, but the partition with the highest welfare is better than the best matching by a factor of about~$2$.
We remark that our construction only uses instances with rational valuations, even though we also exclude irrational competitive ratios higher than $\frac{1}{6+4\sqrt{2}}$.

\begin{restatable}{theorem}{dissolutionbound}\label{thm:dissolutionbound}
	No deterministic online coalition formation algorithm operating under free dissolution has a competitive ratio of more than $\frac{1}{6+4\sqrt{2}}$ for symmetric FHGs.
\end{restatable}

\begin{proof}[Proof sketch]
The crucial idea is to use an algorithm that allegedly beats a competitive ratio of $\frac{1}{6+4\sqrt{2}}$ to construct a sequence of real numbers $(x_i)_{i\in \mathbb N}$ with $x_1 = 1$, $x_i\ge 0$ for $i\ge 2$, and such that for all $i\in \mathbb N$, it holds that
	\begin{equation}\label{eq:Varasequence}
		x_i \ge \beta\left(x_{i+1} + \sum_{j=1}^{i+1}x_j\right)
	\end{equation}
where $\beta > \frac{1}{3+2\sqrt{2}}$.
The proof of Theorem~2 by \citet{Vara11a} for the case of $k=2$ and $f=0$ shows that such a sequence of numbers does not exist.

\begin{figure}
	\centering
	\begin{tikzpicture}
		\node[smallagent] (a) at (0,0) {};
		\node[smallagent] (b) at (6.5,0) {};
		\node at (a) {$a_i$};
		\node at (b) {$b_i$};
		\node at ($(a) + (0,.7)$) {\color{CoalitionColor} $C_i$};
		\node at ($(b) + (0,.7)$) {\color{CoalitionColor} $D_i$};
		
		\draw (a) edge node[midway, fill = white] {$y_i$} (b);
		
		\node[smallagent] (c) at (9.5,0) {};
		\draw (b) edge node[pos =.62, fill = white] {$y_i+j^*\epsilon_i$} (c);
		
		\foreach \i in {a,b}{		
		\node[aux] (l21) at ($(\i) + (240:3)$) {};
		\node[aux] (l22) at ($(\i) + (260:3)$) {};
		\draw (\i) edge node[pos = .65, fill = white] {$y_i+\epsilon_i$} (l21);
		\draw (\i) edge node[pos = .75, fill = white] {$y_i+\epsilon_i$} (l22); 
		\node[rotate = 160] at ($(l21)!.5!(l22)$) {\dots};

		\node[aux] (l11) at ($(\i) + (195:3)$) {};
		\node[aux] (l12) at ($(\i) + (215:3)$) {};
		\draw (\i) edge node[midway, fill = white] {$y_i$} (l11);
		\draw (\i) edge node[midway, fill = white] {$y_i$} (l12); 
		\node[rotate = 115] at ($(l11)!.5!(l12)$) {\dots};
		
		\node[aux] (l31) at ($(\i) + (310:3)$) {};
		\node[aux] (l32) at ($(\i) + (330:3)$) {};
		\draw (\i) edge node[pos = .53, fill = white,xshift = .1cm] (m) {$y_i+k^*\epsilon_i$} (l32); 
		\draw (\i) edge node[pos =.68, fill = white,xshift = .25cm] {$y_i+k^*\epsilon_i$} (l31);
		\draw[shorten <=2cm] (\i) -- (l32);
		\node[rotate = 50] at ($(l31)!.5!(l32)$) {\dots};
		
		\node[rotate = 15] at ($(l22)!.5!(l31)$) {\textbf{\dots}};
		
		\draw[thick,CoalitionColor, fill=CoalitionColor!50, fill opacity=0.2]  \convexpath{\i,l32,l31}{.42cm};
		}

		\draw (a) edge ($(a)-(1,0)$);
		\node at ($(a)-(1.3,0)$) {\dots};
	\end{tikzpicture}
	\caption{Illustration of Phase~$i$ in the construction of the adversarial instance in the proof of \Cref{thm:dissolutionbound}.
	Each star attached to $a_i$ and $b_i$ contains $\ell_i$ leaves.
	\label{fig:dissolutionsketch}}
\end{figure}
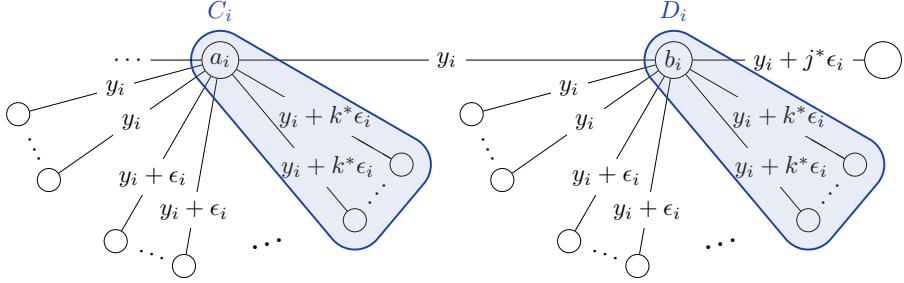

The adversarial instance is established in phases, and in each phase, we determine a new element of a sequence ${(y_i)}_{i\in \mathbb N}$ of rational numbers that satisfies an inequality of the type of Inequality~(\ref{eq:Varasequence}).\footnote{It is easy to eventually transform this sequence to the exact desired form of ${(x_i)}_{i\in \mathbb N}$.}

Throughout the execution of the instance, the algorithm can only maintain a single coalition with positive welfare of $y_i$ containing exactly two agents, say $\{a_i,b_i\}$.
We now illustrate a Phase~$i$ for some fixed $i\in \mathbb N$.
A visualization is provided in \Cref{fig:dissolutionsketch}.
All agents that newly appear have a mutual positive valuation with exactly one of $a_i$ and $b_i$, a valuation of~$0$ for some other agents, and a high negative valuation for most agents, in particular for the other agent in $\{a_i,b_i\}$.
The new agents form ``star'' coalitions with $a_i$ and $b_i$.
In the first part of a stage, we achieve a situation where stars with $\ell_i$ leaves have arrived for both endpoints, where all of their positive valuations are $y_i$.
These are the leftmost stars attached to $a_i$ and $b_i$ in \Cref{fig:dissolutionsketch}.

Then, we let new star coalitions arrive while incrementing their positive valuations by a specifically tailored rational value $\epsilon_i$ in each step.
Eventually, the algorithm has to dissolve $\{a_i,b_i\}$ and form a new coalition with one of these agents and a new agent of valuation $y_i + j^*\epsilon_i$ for some positive integer $j^*$.
This has to happen as otherwise, edges of unbounded weight arrive, which would lead to an unbounded competitive ratio.

In the previous step, i.e., when agents with valuations of $y_i + k^*\epsilon_i$, where $k^* = j^*-1$ were arriving, we had two ``star'' coalitions with $a_i$ and $b_i$, which we now call $C_i$ and $D_i$, respectively.
Then, a version of Inequality~(\ref{eq:Varasequence}) can be established with two differences: (1) instead of $\beta$, we have $2\gamma$, where $\gamma$ is the competitive ratio of our algorithm, and (2) there is an error term dependent on $\epsilon_i$.
For this, we compare $y_i$, i.e., the social welfare of $\{a_i,b_i\}$, with the social welfare of the partition containing $D_i$ and $C_j$ for $1\le j \le i$, where the $C_j$ evolve from earlier phases.
Note that $C_i$ and $D_i$ have a welfare of about $2 (y_i + j^* \epsilon_i)$.

A crucial idea is to control the error terms to be very small in sum by having $\epsilon_i$ decay exponentially for $i$ tending to infinity, while the number of leaves $\ell_i$ grows as $\frac{1-\epsilon_i}{\epsilon_i}$.
This allows to prove Inequality~(\ref{eq:Varasequence}) for $\beta = \gamma + \frac{1}{6+4\sqrt{2}}$.
\end{proof}

\section{Conclusion}

We have studied two different models for online coalition formation in FHGs to maximize social welfare, a goal that does not allow for bounded competitive ratios in the standard model with an adversarial arrival order.
Designing good online coalition formation algorithms is deeply related to designing good online matching algorithms.
It is possible to leverage matching algorithms with a factor~$2$ in welfare loss, while limitations for matching algorithms can be preserved exactly if they hold on the tree domain.

Under random arrival, we have seen that a competitive ratio of $\frac 13$ is asymptotically optimal in both the matching and coalition formation domain.
Hence, by providing a better analysis of a matching algorithm, we do not lose the factor of~$2$ in welfare.
Moreover, in the coalition dissolution model, we determined that the optimal competitive ratio is $\frac1{6+4\sqrt2}$ as compared to the optimal competitive ratio of $\frac1{3+2\sqrt2}$ in the matching domain.
For this, we constructed a new family of worst-case instances showing that the worst-case behavior on the coalition formation domain can be a factor of~$2$ worse than on the matching domain.

An intriguing question is whether forming coalitions larger than~$2$ can be beneficial.
In fact, our paper reinforces the opposing view that matching algorithms exhibit (asymptotically) optimal performance.
Thus, from an algorithmic perspective, larger coalitions are often unnecessary.
In contrast, requiring algorithms to form larger coalitions can be problematic as such algorithms may fail to provide guarantees regarding approximate social welfare.
For example, partitions that include a coalition of size at least three result in a negative welfare for instances on the tree domain.
However, this depends on the presence of large negative valuations.
In contrast, \citet[Theorems~4.4 and~4.5]{FMM+21a} present an algorithm that forms larger coalitions for FHGs with nonnegative valuations.
Nevertheless, the competitive ratio they achieve depends on the range of the involved valuations.
It would be interesting to explore whether this dependency can be eliminated under coalition dissolution or random arrival.

Another future direction is to consider the combination of random arrival and coalition dissolution.
Note that our families of worst-case instances are specifically tailored to their setting and do not challenge algorithms in the respective other setting.
For example, the algorithms developed by \citet{McGr05a} and \citet{BuRo25b} discussed in the beginning of \Cref{sec:freediss} find the optimal partition on stars and double stars (i.e., the worst-case instance in the random arrival setting can be solved optimally under free dissolution, even if agents arrive in an arbitrary order).
Hence, it could well be that algorithms can achieve higher competitive ratios in a combined setting.

\appendix

\section*{Appendix}

In the appendix, we present proofs and further material.

\section{Proofs in Section~\ref{sec:connections}}\label{app:prelimresults} 

In this section, we present the proofs missing in \Cref{sec:connections}.

\subsection{Simple proof of Theorem~\ref{thm:matchingguarantee}} 

In this section, we provide a new proof for \Cref{thm:matchingguarantee} which was first proved by \citet[][Theorem 1]{FKMZ21a}.
The proof by \citet{FKMZ21a} is based on a comparison of the maximum weight matching and the optimal partition by deriving connections of edges not contained in the matching. 
By contrast, we present a simple proof based on a folklore result about matchings that states that the maximum weight of a matching exceeds the weight of a ``uniform fractional'' matching where each edge is fractionally matched with probability $\frac 1n$, if $n$ is the number of vertices.
We combine this with the idea that the social welfare of a coalition in an FHG is twice the weight of the ``uniform fractional'' matching.

\matchingguarantee*

\begin{proof}
	Assume that we are given an FHG $G = (N,\vf)$ and let $\pi^*$ be a partition maximizing social welfare.
	Let $C\subseteq N$ be a coalition and $\mu^*(C)$ be a maximum weight matching on the subgraph of $(N,\vf)$ induced by $C$.
	Recall that we denote the weight of a matching $M$ by $\vf(M)$. 
	
	A folklore theorem in matching \citep[see, e.g.,][Lemma~1]{BuRo25b} says that
	\begin{equation}\label{eq:avgedge}
		\frac 1{|C|}\sum_{\{i,j\}\subseteq C, i\neq j}\vf(i,j) \le \vf(\mu^*(C))\text.
	\end{equation}
	
	We conclude that 
	\begin{align*}
		\SW(\pi^*) &= \sum_{C\in \pi^*}\SW(C) = \sum_{C\in \pi^*} \sum_{i\in C}\sum_{j\in C\setminus \{i\}}\frac {\vf_i(j)}{|C|} = \sum_{C\in \pi^*} \frac 1{|C|}\sum_{\{i,j\}\subseteq C, i\neq j} 2 \vf(i,j) \\
		&\overset{\text{Eq.~(\ref{eq:avgedge})}}
		{\le} \sum_{C\in \pi^*} 2\vf(\mu^*(C))\le 2\vf(\mu^*(N))
		= 2\SW(\mu^*(N))\text.
	\end{align*}
	
	The second line uses that each valuation occurs twice in a symmetric game, once for each endpoint.
	The second-to-last line uses that $\bigcup_{C\in \pi^*} \mu^*(C)$ is a matching on $N$, and, therefore, its weight is bounded by the maximum weight matching of $N$.
	The last line uses that each nonsingleton coalition $C = \{i,j\}$ in $\mu^*(N)$ consists of two agents, i.e., $\SW(C) = \uf_i(C) + \uf_j(C) = \frac 12 \vf(i,j) + \frac 12 \vf(i,j) = \vf(i,j)$.
\end{proof}

\subsection{Matching Limitations on the Tree Domain}

In this section, we provide the proof for our result that lets us transfer negative results for online matching algorithms on the tree domain to the coalition formation domain.

\treenegative*

\begin{proof}
	We show a proof by contraposition.
	Assume a $c$-competitive online coalition formation algorithm $\alg$ for symmetric FHGs exists.
	We construct a $c$-competitive algorithm $\alg'$ on the tree domain that never forms a coalition of size three or more.
	To this end, let $\alg'$ simulate $\alg$, i.e., whenever a new agent and her valuations are revealed to $\alg'$, it feeds the same input to $\alg$.
	Then, $\alg'$ observes the output of $\alg$.
	If the new agent is in a coalition of size two with positive social welfare, then $\alg'$ forms the same coalition.
	In all other cases, $\alg'$ puts the new agent into a singleton coalition.
	Additionally, if $\alg$ dissolves a coalition in the coalition dissolution setting, then $\alg'$ also dissolves the matched pair from this coalition if necessary.
	In particular, $\alg'$ only returns (random) matchings and, therefore, is a matching algorithm.
	
	On the tree domain, $\alg'$ achieves at least as high (expected) welfare as $\alg$ because the large negative valuations make every coalition of size more than two have negative social welfare.
	Consequently, every coalition of size at least~$3$ achieves less welfare than when it was dissolved into singleton coalitions (or pairs of positive valuation).
	Thus, $\alg'$ is $c$-competitive on the tree domain against all possible partitions and, therefore, in particular, against all matchings.
\end{proof}

\section{Further Material for Section~\ref{sec:randomarrival}}\label{app:randomarrival} 

In this section, we present further details for the random arrival setting.

\subsection{Proof of Theorem~\ref{thm:online_mwm_unknown_n}}

In this section, we consider the performance of the algorithm by \citet{EFGT22a}, when the number of agents is unknown and $k$ is set to $3$.
\onlineMWM*

\begin{proof}
	Consider Algorithm~$1$ as defined by \citet{EFGT22a}.
	We refer to this algorithm as $\alg$. Note that the algorithm is parameterized by a positive integer $k$.

	Consider an arbitrary FHG $G = (N,\vf)$. 
	Let $\mu^*$ be a maximum weight matching and $\alg(G)$ be the matching computed by $\alg$.
	In the proof of their Theorem~3.1, \citet{EFGT22a} obtain the following inequality for all $k\ge 3$:
	\begin{equation*}
		\frac{\EVA{N}\left[\vf(\alg(G))\right]}{\vf(\pi^*(G))} \ge \frac 13 + \frac{k^2}{n^2} - \frac{4k^3}{3n^3} - \mathcal O\left(\frac 1n\right)\text.
	\end{equation*}
	Setting $k = 3$, this implies that the competitive ratio of $\alg$ is at least
	\begin{equation*}
		\inf_G\frac{\EVA{N}\left[\vf(\alg(G))\right]}{\vf(\pi^*(G))} \ge \frac 13 - \mathcal O\left(\frac 1n\right)\text.\qedhere
	\end{equation*}
\end{proof}

\subsection{Optimal Coalition Formation Algorithm}\label{app:randomarrival:optalg}

In this section, we provide the proofs of auxiliary lemmas in the proof of \Cref{thm:CofoRandArr}, restated as follows.

\CofoRandArr*

The statement of \Cref{lem:twothirdsmatched} is similar to the statement of Lemma~3.2 by \citet{EFGT22a}.
While the proofs are quite different, they both rely on showing that the probability of being unavailable can be captured by the following recursive formula:
For all $k,t\in \mathbb N$, $t\ge k+1$, define 
	\begin{equation}\label{eq:EFGTrecursionformula}
		p(k,k):=0 \text{ and }p(k,t):=\frac2t+\frac{t-3}t\cdot p(k,t-1)\text.
	\end{equation}

\citet{EFGT22a} show that this formula resolves as captured in the following lemma, see the end of the proof of their Lemma~3.2.

\begin{lemma}[\citet{EFGT22a}]\label{lem:recursion}
	For $k,t\in \mathbb N$ with $k\ge 3$ and $t\ge k$, it holds that 
	\begin{equation*}
		p(k,t) = \frac23\left(1-\frac{(t-3)!\cdot k!}{t!\cdot(k-3)!}\right)\text.
	\end{equation*}
\end{lemma}

Note that this implies that the probability of leaving $a_t$ unmatched in \Cref{line:X} of \Cref{algorithm:onlinemwm} is $1-p(k,t-1)$.
This already suggest of a relationship of the resolution of the recursion with the probability of availability.
In fact, we now apply it to prove our first key lemma. 

\twothirdsmatched*

\begin{proof}
	Consider any realization $\tilde{\agS}$ of $\agS_t$ and any agent $a\in \tilde{\agS}$.
	By \Cref{lem:recursion}, it suffices to prove that
	\begin{equation}\label{eq:EFGTrecursion}
		\prob\left(a\in \agS_t\setminus A_t\mid \agS_t=\tilde{\agS}\land a\in\tilde{\agS}\right)
		=p(k,t)\text.
	\end{equation}
	
	We prove this 
	by induction over $t$.
	Clearly, since the arrival of the first $k$ agents does not alter the set $A_t$, we have that $N_k = A_t$ and, therefore, 
	$\prob\left[a\in \agS_t\setminus A_t\mid \agS_t=\tilde{\agS}\land a\in\tilde{\agS}\right]= 0 = p(k,k)$. 
	
	Now assume that $t>k$.
	We partition the event that $a$ is unavailable (i.e., not in $A_t$) after the $t$th agent has been processed 
	and given that $\agS_t=\tilde{\agS}$, into the following disjoint events:
	\begin{enumerate}[label=(\roman*)]
		\item $a=a_t$,
		\item $a=p_t$ and $a\neq a_t$, and 
		\item $a=x$ for some $x\in\tilde{\agS}\setminus e_t$.
	\end{enumerate}
	Event (i) occurs with probability $\frac1t$.
	For (ii), we compute
	\begin{align*}
		\prob\left(a=p_t,a\neq a_t\right)
		&=\sum_{x\in\tilde{\agS}\setminus\{a\}}\prob\left(a=p_t\mid x=a_t\right)\cdot\prob\left(x=a_t\right)
		=\sum_{x\in\tilde{\agS}\setminus\{a\}}\fracmat_t(x,a)\frac1t\\
		&=\sum_{x\in\tilde{\agS}\setminus\{a\}}\fracmat_t(a,x)\frac1t
		=\left(1-\fracmat_t(a,a)\right)\frac1t\text{.}
	\end{align*}
	Hence, the probability for event (iii) is $1-\frac1t-\left(1-\fracmat_t(a,a)\right)\frac1t=\frac{t-2+\fracmat_t(a,a)}t$.
	We get
	\begin{align*}
		\prob\left(a\in \agS_t\setminus A_t \mid \agS_t=\tilde{\agS} \right.&\left.\land a\in\tilde{\agS}\right)
		=\prob\left(a\in \agS_t\setminus A_t \mid \agS_t=\tilde{\agS}\land a\in\tilde{\agS}\land a=a_t\right)\frac1t\\
		&+\prob\left(a\in \agS_t\setminus A_t \mid \agS_t=\tilde{\agS}\land a\in\tilde{\agS}\land a\notin e_t\right)
		\frac{t-2+\fracmat_t(a,a)}t\\
		&+\prob\left(a\in \agS_t\setminus A_t \mid \agS_t=\tilde{\agS}\land a\in\tilde{\agS}\land a=p_t\land a\neq a_t\right)\left(1-\fracmat_t(a,a)\right)\frac1t\text{.}
	\end{align*}
	
	Assume that $a=a_t$.
	If $p_t \neq a_t$, then $a_t$ is unavailable after being processed if and only if $p_t$ was available at the arrival of $a_t$, i.e., after the $(t-1)$st agent was processed.
	By the induction hypothesis for $\agS_{t-1}=\tilde{\agS}\setminus\{a\}$, agent $p_t$ was made unavailable with probability $p(k,t-1)$ and is, therefore, available with probability $1-p(k,t-1)$.
	Moreover, if $p_t = a_t$, then $a_t$ becomes unavailable if they are removed in \cref{line:X}.
	By definition this happens with probability $\frac 13 + \frac{2(t-4)!k!}{3(t-1)!(k-3)!} = 1-p(k,t-1)$ as well.
	Hence, in any case, $a$ becomes unavailable with probability $1-p(k,t-1)$ after being processed.
	
	If $a$ is neither $a_t$ nor $p_t$, then $a$ is unavailable at the end of iteration $t$ if they are unavailable at the end of iteration $t-1$.
	Furthermore, if they are available at the end of iteration $t-1$, they can be made unavailable in \cref{line:remove} of the algorithm.
	This happens with probability $\frac{\fracmat_t(a,a)}{t-2+\fracmat_t(a,a)}$.
	Using the induction hypothesis, we compute that $a$ is unavailable by time $t$ with probability
	$p(k,t-1)+(1-p(k,t-1))\frac{\fracmat_t(a,a)}{t-2+\fracmat_t(a,a)}$.

	Finally, if $a=p_t$, then $a$ is always unavailable at the end of iteration~$t$. Indeed, if $a$ was available during the iteration, it will be matched to $a_t$ upon their arrival.

	Putting it all together, we get
	\begin{align*}
		&\prob\left(a\in \agS_t\setminus A_t \mid \agS_t=\tilde{\agS}\land a\in\tilde{\agS}\right)
		=(1-p(k,t-1))\frac1t+\\ 
		&\left(p(k,t-1)+
		\left(1-p(k,t-1)\right)\frac{\fracmat_t(a,a)}{t-2+\fracmat_t(a,a)}\right)\frac{t-2+\fracmat_t(a,a)}t 
		+\left(1-\fracmat_t(a,a)\right)\frac1t\\
		&=(1-p(k,t-1))\frac1t+\frac1t\left(\left(t-2+\fracmat_t(a,a)\right)p(k,t-1)+
		\left(1-p(k,t-1)\right)\fracmat_t(a,a)\right)+\left(1-\fracmat_t(a,a)\right)\frac1t\\
		&=\frac1t\Big(1-p(k,t-1)+\left(t-2+\fracmat_t(a,a)\right)p(k,t-1)+
		\left(1-p(k,t-1)\right)\fracmat_t(a,a)+1-\fracmat_t(a,a)\Big)\\
		&=\frac1t\Big(2-p(k,t-1)+\left(t-2\right)p(k,t-1)+\fracmat_t(a,a)
		p(k,t-1)+
		\left(1-p(k,t-1)\right)\fracmat_t(a,a)-\fracmat_t(a,a)\Big)\\
		&=\frac1t\left(2-p(k,t-1)+\left(t-2\right)p(k,t-1)\right)\\
		&=\frac2t+\frac{t-3}t\cdot p(k,t-1)\\
		&\overset{\text{Eq.~(\ref{eq:EFGTrecursionformula})}}=p(k,t)\text{.}
	\end{align*}
	This concludes the proof of \Cref{eq:EFGTrecursion} and thus the entire proof.
\end{proof}

Next, we prove the lemma concerning the monotonicity property when running \Cref{algorithm:onlinemwm}.

\monotonic*

\begin{proof}
	Let $t\ge k + 1$.
	Consider any subset $\tilde{\agS}\subseteq \agS$ that realizes as $\agS_t$ and an agent $a\in \tilde{\agS}$.
	We set $\vf(e_t) = 0$ if $e_t = \{a_t\}$, i.e., $e_t$ consists only of a singleton.
	We then have
	\begin{align*}
		&\EV\left[\vfd(e_t)\mid \agS_t=\tilde{\agS}\land a\in\tilde{\agS}\right]\\
		&=\sum_{x\in\tilde{\agS}}\EV\left[\vfd(e_t)\mid \agS_t=\tilde{\agS}\land a\in\tilde{\agS}\land a_t = x\right]\cdot\prob\left(a_t = x\mid \agS_t=\tilde{\agS}\land a\in\tilde{\agS}\right)\\ 
		&=\sum_{x\in\tilde{\agS}}\EV\left[\vfd(e_t)\mid \agS_t=\tilde{\agS}\land a\in\tilde{\agS}\land x=a_t\right]\cdot\frac1t\\
		&=\frac1t\sum_{x\in\tilde{\agS}}\sum_{x'\in\tilde{\agS}}\EV\left[\vfd(e_t)\mid \agS_t=\tilde{\agS}\land a\in\tilde{\agS}\land x=a_t\land x'=p_t\right]\\
		&\phantom{=}\cdot\prob\left(x'=p_t\mid \agS_t=\tilde{\agS}\land a\in\tilde{\agS}\land x=a_t\right)\\
		&=\frac1t\sum_{x\in\tilde{\agS}}\sum_{x'\in\tilde{\agS}}\vfd(x,x')\cdot\fracmat_t(x,x')\\
		&=\frac2t\EV\left[\wf(\fracmat_t)\right]
	\end{align*}
	
	In the first and third equation, we use the law of total probability.
	The last equality follows because with the two summations we count every pair of vertices twice.

	Next, we consider the expected competitive ratio.
	It holds
	\begin{align*}
		&\EVA{N}\left[\SW_G\left(\alg(G)\right)\right]
		=\sum_{t=k+1}^n\EV\Big[\I\left[p_t\in A_{t-1}\right]\cdot \vfd(e_t)\Big]\\ 
		&=\sum_{t=k+1}^n\EV\Big[\I\left[p_t\in A_{t-1}\right]\Big]\cdot\EV\left[\vfd(e_t)\right]
		=\sum_{t=k+1}^n\prob\left(p_t\in A_{t-1}\right)\cdot\EV\left[\vfd(e_t)\right]
	\end{align*}
	
	We use that the random variables $\I\left[p_t\in \agS_t\cap A_t\right]$ and $\vfd(e_t)$
	are independent of each other. 
	Indeed, the MWFM computed in \cref{ln:independenceMWFM} of \Cref{algorithm:onlinemwm} was chosen independent of $A$ (which at this points was equal to $A_{t-1}$).

	Moreover, by \Cref{lem:twothirdsmatched}, we know that
	\begin{equation*}
		\prob\left(p_t\in A_{t-1}\right) = 1 - p(k,t-1)\text.
	\end{equation*}
	Hence,
	\begin{align*}
		&\EVA{N}\left[\SW_G\left(\alg(G)\right)\right]
		=\sum_{t=k+1}^n\prob\left(p_t\in A_{t-1}\right)\cdot\EV\left[\vfd(e_t)\right]
		=\sum_{t=k+1}^n\left(1-p(k,t-1)\right)\frac2t\EV\left[\wf(\fracmat_t)\right]\\
		&\ge\sum_{t=k+1}^n\left(1-p(k,t-1)\right)\frac2t\EV\left[\wf(\fracmat'_t)\right]
		\ge\sum_{t=k+1}^n\left(1-p(k,t-1)\right)\frac2t\EV\left[\wf'(\fracmat'_t)\right]\\
		&=\sum_{t=k+1}^n\prob\left(p_t\in A_{t-1}\right)\cdot\EV\left[\vfd'(e'_t)\right]
		=\EVA{N}\left[\SW_{G'}\left(\alg(G')\right)\right]\text.
	\end{align*}
	There, $\fracmat'_t$ and $e'_t$ denote the constructed MWFM and edges when executing \Cref{algorithm:onlinemwm} for $G'$.
	In the first inequality, we use that $\fracmat_t$ is an MWFM for $G[N_t]$ and hence achieve a higher weight in $G$ than $\fracmat'_t$.
	In the second inequality, we use that $\vf(x,y)\ge \vf'(x,y)$ for all agents $x,y\in N$.
	This completes the proof of the lemma.
\end{proof}

Our next lemma shows that the desired performance is achieved on a special set of instances where valuations between coalitions of some optimal partition are negative.

\IslandInstances*

\begin{proof}
	Let $G = (N,\vf)$ be an FHG that satisfies the properties of the statement of the lemma together with a uniformly random arrival order $\sigma$.
	
	For $t\le 3$, we set $e_t = \{a_t\}$.
	Moreover, as in the proof of \Cref{lem:monotonic}, we set $\vf(e_t) = 0$ if $e_t = \{a_t\}$. 
	Now, we know from the proof of \Cref{lem:monotonic} that
	\begin{align*}
		\EVA{N}\left[\SW\left(\alg(G)\right)\right]
		&=\sum_{t=4}^n\prob\left(p_t\in A_{t-1}\right)\cdot\EV\left[\vf(e_t)\right]
		=\sum_{t=4}^n\left(1-p(3,t-1)\right)\cdot\EV\left[\vf(e_t)\right]\text.
	\end{align*}
	We observe that for all $t\ge 4$, it holds that
	\begin{equation*}
		1 - p(3,t-1) = 1-\frac23\left(1-\frac{(t-4)!\cdot6}{(t-1)!}\right)
		>\frac13\text.
	\end{equation*}
	
	Consequently, 
	\begin{align*}
		\EVA{N}\left[\SW\left(\alg(G)\right)\right]
		&>\frac13\sum_{t=4}^n\EV\left[\vfd(e_t)\right]\\
		&=\frac13\sum_{\underset{\sigma^{-1}(x)\ge 4}{x\in \agS}}\EV\left[\vfd(e_{\sigma^{-1}(x)})\right]\\
		&=\frac13\sum_{\underset{\sigma^{-1}(x)\ge 4}{x\in \agS}}\EV\left[\vfd(e_{\sigma^{-1}(x)})\mid\sigma^{-1}(x)\ge 4\right]\\
		&=\frac13\sum_{x\in \agS}\EV\left[\vfd(e_{\sigma^{-1}(x)})\mid\sigma^{-1}(x)\ge 4\right]\cdot\prob\left(\sigma^{-1}(x)\ge 4\right)\\
		&=\frac13\sum_{x\in \agS}\EV\left[\vfd(e_{\sigma^{-1}(x)})\mid\sigma^{-1}(x)\ge 4\right]\cdot\frac{n-3}n\\
		&=\frac{n-3}{3n}\sum_{C\in\pi^*}\sum_{x\in C}\EV\left[\vfd(e_{\sigma^{-1}(x)})\mid\sigma^{-1}(x)\ge 4\right]\\
		&=\frac{n-3}{3n}\sum_{C\in\pi^*}\sum_{x\in C}\EV\left[\vfd(e_{\sigma^{-1}(x)})\mid \agS_{\sigma^{-1}(x)-1}\cap C\neq\emptyset\land\sigma^{-1}(x)\ge 4\right]\\
		&\phantom{=}\cdot\prob\left(\agS_{\sigma^{-1}(x)-1}\cap C\neq\emptyset\mid\sigma^{-1}(x)\ge 4\right)\\
		&\ge\frac{n-3}{3n}\sum_{C\in\pi^*}\sum_{x\in C}\EV\left[\vfd(e_{\sigma^{-1}(x)})\mid \agS_{\sigma^{-1}(x)-1}\cap C\neq\emptyset\land\sigma^{-1}(x)\ge 4\right]\cdot\frac{\lvert C\rvert-1}{\lvert C\rvert}\\
		&=\frac{n-3}{3n}\sum_{C\in\pi^*}\sum_{t\in[n]:\sigma(t)\in C}\EV\left[\vfd(e_t)\mid \agS_{t-1}\cap C\neq\emptyset\land t\ge 4\right]\cdot\frac{\lvert C\rvert-1}{\lvert C\rvert}\text.
	\end{align*}
	
	For the inequality at second-to-last line, observe that $\prob\left(\agS_{\sigma^{-1}(x)-1}\cap C\neq\emptyset\right) = \frac{|C|-1}{|C|}$ because every agent in $C$ does not appear as the first agent among the agents in $C$ with that probability.
	Moreover, as having agents present already can only increase the probability, this quantity can only become larger when conditioning on $\sigma^{-1}(x)\ge 4$.
	
	Let $t\in [n]$ with $\sigma(t) \in C$. 
	We now take a closer look at $\EV\left[\vfd(e_t)\mid \agS_{t-1}\cap C\neq\emptyset\land t\ge 4\right]$.
	Essentially, this is the expected weight of the candidate edge if the arriving agent $\sigma(t)=a_t$ is not the first of $\pi^*(a_t)$ to arrive.
	We define the fractional matching $\fracmat_{C\cap \agS_t}$ on $G[C\cap \agS_t]$ by setting 
	\begin{equation}\label{eq:fracmatisland}
		\fracmat_{C\cap \agS_t}(x,x') = \fracmat_t(x,x')
	\end{equation}
	for all $x,x'\in\agS_t$ with $x,y\in C\in\pi^*$.
	
	By our assumption on the structure of the instance, $\fracmat_t$ cannot assign positive weight on edges between agents from different coalitions in $\pi^*$.
	Hence, for all $t\ge 4$, $\fracmat_{C\cap \agS_t}$ is a well-defined fractional matching on $G[C\cap \agS]$ and we have
	\begin{equation*}
		\vf(\fracmat_t)=\sum_{C\in\pi^*}\vf(\fracmat_{C\cap \agS_t})\text.
	\end{equation*}
	
	Recall that $t\in [n]$ with $\sigma(t) \in C$.
	Now, similar as in the proof of \Cref{lem:monotonic}, only restricted to a single coalition, we compute: 
	
	\begin{align}\label{eq:auxeqone}
		&\EV\left[\vfd(e_t)\mid \agS_{t-1}\cap C\neq\emptyset\land t\ge 4\right]\nonumber\\
		&=\sum_{x\in\agS_t\cap C}\EV\left[\vfd(e_t)\mid \agS_{t-1}\cap C\neq\emptyset\land t\ge 4\land a_t = x\right]\cdot\prob\left(a_t = x\mid \agS_{t-1}\cap C\neq\emptyset\land t\ge 4\right)\nonumber\\
		&=\sum_{x\in\agS_t\cap C}\EV\left[\vfd(e_t)\mid \agS_{t-1}\cap C\neq\emptyset\land t\ge 4\land x=a_t\right]\cdot\frac1{\lvert\agS_t\cap C\rvert}\nonumber\\
		&=\frac1{\lvert\agS_t\cap C\rvert}\sum_{x\in\agS_t\cap C}\sum_{x'\in\agS_t\cap C\setminus\{x\}}\EV\left[\vfd(e_t)\mid \agS_{t-1}\cap C\neq\emptyset\land t\ge 4\land x=a_t\land x'=p_t\right]\nonumber\\
		&\phantom{=}\cdot\prob\left(x'=p_t\mid \agS_{t-1}\cap C\neq\emptyset\land t\ge 4\land x=a_t\right)\nonumber\\
		&=\frac1{\lvert\agS_t\cap C\rvert}\sum_{x\in\agS_t\cap C}\sum_{x'\in\agS_t\cap C\setminus\{x\}}\vfd(x,x')\cdot\fracmat_t(x,x')\nonumber\\
		&\overset{\text{Eq.~(\ref{eq:fracmatisland})}}{=}\frac1{\lvert\agS_t\cap C\rvert}\sum_{x\in\agS_t\cap C}\sum_{x'\in\agS_t\cap C\setminus\{x\}}\vfd(x,x')\cdot\fracmat_{C\cap \agS_t}(x,x')\nonumber\\
		&=\frac2{\lvert C\cap \agS_t\rvert}\EV\left[\vf(\fracmat_{C\cap \agS_t})\right]\text.
	\end{align}

	Moreover, we know that in general, for any subgraph of $G$ induced by agent set $W\subseteq \agS$, an MWFM $\fracmat_W$ on $G[W]$ is at least as good as matching all edges with equal fractions, i.e.,
	\begin{equation}\label{eq:auxeqtwo}
		2\cdot\vf(\fracmat_{W})
		\ge\sum_{x\in W}\sum_{x'\in W\setminus\{x\}}\frac1{|W|-1}\vfd(x,x')
	\end{equation}
	
	For a subset $W\subseteq \agS$, let $U(W)$ denote the uniform distribution over ordered pairs of distinct agents in $W$, i.e., edges in the subgraph induced by $W$. 
	Combining the previous insights, we get
	\begin{align*}
		&\EV\left[\vfd(e_{t})\mid \agS_{t-1}\cap C\neq\emptyset\land t\ge 4\right]
		\overset{\text{Eq.~(\ref{eq:auxeqone})}}{=}\frac2{\lvert C\cap \agS_{t}\rvert}\EV\left[\vf(\fracmat_{C\cap \agS_{t}})\right]\\
		&\overset{\text{Eq.~(\ref{eq:auxeqtwo})}}{\ge}\frac1{\lvert C\cap \agS_{t}\rvert}\EV\left[\sum_{x\in C\cap \agS_{t}}\sum_{x'\in C\cap \agS_{t}\setminus\{x\}}\frac1{\lvert C\cap \agS_{t}\rvert-1}\vfd(x,x')\right]\\
		&=\frac1{\lvert C\cap \agS_{t}\rvert}\frac{(\lvert C\cap \agS_{t}\rvert-1)\lvert C\cap \agS_{t}\rvert}{\lvert C\cap \agS_{t}\rvert-1}\EV_{e\sim U(C)}\left[\vfd(e)\right]\\
		&=\EV_{e\sim U(C)}\left[\vfd(e)\right]\text.
	\end{align*}
	
	There, the second-to-last equality holds because $C\cap \agS_{t}$ is a subset of $C$ drawn uniformly at random,
	and hence drawing a uniformly random valuation from an ordered pair in $C\cap \agS_{t}$
	is equivalent to drawing a random valuation from an ordered pair in $C$ directly.
	All in all,
	\begin{align*}
		\EVA{N}\left[\SW\left(\alg(G)\right)\right]
		&\ge\frac{n-3}{3n}\sum_{C\in\pi^*}\sum_{t\in[n]:\sigma(t)\in C}\EV\left[\vfd(e_{t})\mid \agS_{t-1}\cap C\neq\emptyset\land t>3\right]\cdot\frac{\lvert C\rvert-1}{\lvert C\rvert}\\
		&\ge\frac{n-3}{3n}\sum_{C\in\pi^*}\sum_{t\in[n]:\sigma(t)\in C}\EV_{e\sim U(C)}\left[\vfd(e)\right]\cdot\frac{\lvert C\rvert-1}{\lvert C\rvert}\\
		&=\frac{n-3}{3n}\sum_{C\in\pi^*}\lvert C\rvert\cdot\EV_{e\sim U(C)}\left[\vfd(e)\right]\cdot\frac{\lvert C\rvert-1}{\lvert C\rvert}\\
		&=\frac{n-3}{3n}\sum_{C\in\pi^*}\EV_{e\sim U(C)}\left[\vfd(e)\right]\cdot\left(\lvert C\rvert-1\right)\\
		&=\frac{n-3}{3n}\sum_{C\in \pi^*}\SW\left(C\right)\\
		&=\frac{n-3}{3n}\SW\left(\pi^*\right)\text.
	\end{align*}
	This proves the claimed competitive ratio:
	\begin{equation*}
		\frac{\EVA{N}\left[\SW\left(\alg(G)\right)\right]}{\SW\left(\pi^*\right)}
		\ge\frac{n-3}{3n}
		=\frac13-\frac1n
	\end{equation*}
\end{proof}

\subsection{Limitations of Algorithm Based on Integral Matchings} 
\label{app:518}

In this section, we consider \Cref{algorithm:onlinemwm} for the case where we modify \cref{ln:independenceMWFM} to select an MWM instead of an MWFM\@.
We select an MWM that only depends on the set of present agents by applying minuscule perturbations of the valuations that act as a tie-breaking mechanism among integral matchings.
This is essentially analogous to the online matching algorithm by \citet[Algorithm~1]{EFGT22a} for the matching setting when the number of agents is known.\footnote{Our next result also holds for the version of their algorithm for the case when $k = 3$, which is optimal for online matching with an unknown number of agents, cf. \Cref{thm:online_mwm_unknown_n}.}
However, another difference is that they want the maximum weight matchings to be perfect and, therefore, need an even number of vertices.
They achieve this by deleting a uniformly random vertex whenever the number of vertices is odd.
We do not face this issue as we deal with the case of being unmatched by making vertices unavailable based on the probabilities $\fracmat(x,x)$ for an agent $x$ and a fractional matching $\fracmat$.

\begin{proposition}\label{prop:518}
	Let $\alg$ be \Cref{algorithm:onlinemwm} where we modify \cref{ln:independenceMWFM} to select an MWM as described above.
	$\alg$ has a competitive ratio under random arrival of at most $\frac{5}{18} < \frac 13$, even for simple symmetric FHGs.
\end{proposition}
\begin{proof}
	Let $\alg$ be the algorithm from the statement of the proposition.
	Let $k\in \mathbb N$.
	We define an FHG $G_k = (N,\vf)$ with $n = 3k$ agents.
	Let $N = \{a_i,b_i,c_i\colon i\in [k]\}$.
	The valuations are given as $\vf(a_i,b_i)=\vf(a_i,c_i)=\vf(b_i,c_i)=1$ for all $i\in [k]$ and all other valuations, i.e., the valuations across triplets of agents, are set to~$0$.
	Hence, the underlying graph is a disjoint union of triangle graphs.
	Clearly, this constitutes a simple symmetric FHG\@.

	Let now $i\in[k]$ be arbitrarily fixed, i.e., we
	look at a single triangle.
	Assume that agent $a_i$ arrives at time $t_a$, agent $b_i$
	arrives at time $t_b$, and $c_i$ arrives at time $t_c$.
	Without loss of generality, assume that $t_a<t_b<t_c$.
	In the algorithm,
	the edge weights are randomly perturbed to
	always guarantee a unique MWM\@.
	When $a_i$ arrives, the candidate edge $e_{t_a}$ which may be
	added by the algorithm will have weight $0$ (or $a_i$ is left unmatched).
	When $b_i$ arrives, the candidate edge $e_{t_b}$ will
	be the edge $\{a_i,b_i\}$ and hence have weight $1$.
	When $c_i$ arrives, the random perturbation leads to a $\frac13$ chance the
	candidate edge $e_{t_c}$ will be $\{a_i,c_i\}$,
	a $\frac13$ chance it will be $\{b_i,c_i\}$, and
	a $\frac13$ chance it will have weight $0$ (or that the algorithm lets $c_i$ unmatched in an MWM) because $\{a_i,b_i\}$
	is the heaviest of the three edges in the triangle.

	Moreover, whether an edge is included also depends on whether the other agent in this edge is available.
	This probability is given as in the formula of \Cref{lem:twothirdsmatched}.\footnote{The proof of the lemma works when replacing the MWFM by an MWM in every occurrence.}
	Hence, in the limit when $k$ (and, therefore, $n$) tends to infinity, the partner of the candidate edge is available with probability $\frac 13$.
	Thus, in the limit, the expected social welfare of
	the algorithm's output obtained by $a_i$, $b_i$, and $c_i$ is
	\begin{equation*}
		\frac13\left(0+1+\frac23\cdot1\right)
		=\frac59\text.
	\end{equation*}
	Hence, we have that $\lim_{k\to \infty}\EV(\alg(G_k)) = \frac 59 k$.
	
	However, we have that $\pi^*(G_k) = \{\{a_i,b_i,c_i\}\colon i\in [k]\}$ with $\SW(\pi^*(G_k)) = 2k$.
	Hence, 
	$$c_{\alg} \le \lim_{k\to \infty}\frac{\EV(\alg(G_k))}{\SW(\pi^*(G_k))} = \frac {\frac 59k}{2k} = \frac 5{18}\text.$$
	This completes the proof.
\end{proof}

The reason why \Cref{algorithm:onlinemwm} does not run into the same issues is that selecting an MWFM allows to match each of the three edges within a triangle with probability $\frac 12$.
Hence, the social welfare achieved by the third agent in a triangle is $1$ and not $\frac 23$, leading to a higher welfare.
Moreover, note that the instances constructed in \Cref{prop:518} also show that \Cref{algorithm:onlinemwm} can be at most $\frac 13$-competitive, i.e., they determine its asymptotic competitive ratio.
This already hints at the much stronger result of \Cref{cor:RAcofoBound}.

\subsection{Proof of Theorem~\ref{thm:negativeresult_MWMrand}}\label{app:randomarrival:limitmain}

In this section, we provide the proof of \Cref{thm:negativeresult_MWMrand}.
To make it accessible more quickly, we defer the proofs of intermediary lemmas to \Cref{app:randomarrival:limit}.

The first part of our proof concerns showing that a good competitive ratio on a star is essentially equivalent to matching the maximum weight edge with a high probability.
This is similar to the conversion of the problem from a cardinal to an ordinal setting as performed by \citet{EFGT22a}.
We then want bounds for the probability of matching any edge in a star instance, only dependent on the already arrived vertices.
While \citet{EFGT22a}, inspired by \citet{CDFS19a}, carry out such a step by applying an infinite version of Ramsey's theorem, we perform a direct computation of the probabilities using induction.
Still, our proof is quite different from both of these as we have the interplay of two qualitatively different sets of adversarial instances.

\MWMrandNegative*

\begin{proof}
In the following proof, we assume that all algorithms are randomized and operate under random arrival.
Assume for contradiction that there is an online matching algorithm 
on the tree domain with a competitive ratio of $\ctree > \frac 13$. Without loss of generality, we can assume that $\ctree$ is rational. Otherwise, we replace it with any rational number in the open interval $(\frac 13, \ctree)$. 
	Define $\epsilon := \frac 13 \left(\ctree-\frac 13\right)$.
	Note that $\epsilon > 0$ and $\epsilon$ is rational.

Let $I, J \subseteq \mathbb{N}$ with $|I|, |J| < \infty$, $I\cap J = \emptyset$ and $I\neq \emptyset$, i.e., they are finite and disjoint, and $I$ is nonempty. 
We design a family of instances with $n = 2 + \lvert I \rvert + \lvert J \rvert$ agents based on two symmetric valuation functions, one for stars and one for bi-stars, dependent on $I, J$.
Additionally, the instance depends on a value for weights of negative edges, parameterized by $x$.
Given such $I$ and $J$, we define $\bstarmax := \max(I \cup J)$, i.e., $\bstarmax$ is the largest number in $I \cup J$. 
We arbitrarily select an integer $x > \bstarmax + 2$. 
Let $N = \{a, b\} \cup \{d_i \colon i \in I\} \cup \{d_j\colon j \in J\}$ be the set of agents.

First, we define a \emph{star instance} $S^{x}_{I,J}$ by setting the following symmetric valuations:\footnote{We omit references to parameters from the names of the valuation functions to avoid overloading notation.}
For all $i \in I$, we set $\vfd(a,d_i) = \vb^i$. 
All remaining valuations are set to $-\vb^x$.
We set $\starmax := \max I$, i.e., the edge of maximum weight is $\{a,d_{\starmax}\}$ with a weight of $\vb^{\starmax}$.
Note that $\starmax > 0$ as $I\neq \emptyset$.

Moreover, we define a \emph{bi-star instance} $B^{x}_{I,J}$ with the following symmetric valuations:
Recall that $\bstarmax = \max(I \cup J)$. 
For all $i \in I$ and $j\in J$, we set $\vfd(a,d_i) = \vb^i$ and $\vfd(b,d_j) = \vb^j$.
We set $\vfd(a,b) = 
		\vb^{\bstarmax + 1}$.
		Finally, all remaining valuations are set to $-\vb^x$.
		Note that the pair $\{a,b\}$ has the highest valuation of $\vb^{\bstarmax + 1}$.
Note that, since $\epsilon$ is rational, all valuations in star and bi-star instances are rational.

Hence, given the same set of parameters, a star and bi-star instance only differ with respect to the valuations of $b$ with $a$ and agents in $\{d_j\colon j\in J\}$.
We denote the set of all star instances with any permissible parameter combination of $I$, $J$, $x$, and $\epsilon$ as $\starins$. 
Similarly, we denote the set of all bi-star instances as $\biins$.

Note that the algorithm can only distinguish star and bi-star instances once $a$ and $b$ have arrived in a bi-star instance.
In fact, once $a$ has arrived in a star instance, or one of $a$ and $b$ has arrived in a bi-star instance, an algorithm sees the star with one of these agents.
However, all other agents, and in particular $b$ if we are in a star instance, are only connected by large constant negative valuations and are indistinguishable.
Furthermore, the optimal matching for star instances matches $\{a, d_{\starmax}\}$ and leaves all other agents alone with a social welfare of $\vb^{\starmax}$.
Similarly, in bi-stars, the optimal matching matches $\{a, b\}$ and leaves all other agents as singletons with a social welfare of $\vb^{\bstarmax + 1}$.

Additionally, by the choice of $x$, both types of instances belong to the tree domain.
Indeed, positive valuations are $\vb^i$ for some $i \le x - 2$ and occur at most once each.
Moreover, as $\ctree\in (\frac 13, 1)$, it holds that $\epsilon = \frac 13(\ctree-\frac 13) \le \frac 12$.
Hence, we have that the sum of positive valuations is at most $\sum_{i = 1}^{x-2}\vb^i \le \vb^{x-1} < \vb^x$.

Given a fixed algorithm, 
we want to find a relationship between its competitive ratio and the probability of matching the highest edge in star and bi-star instances.
We say that an algorithm is \emph{$c$-competitive for matching the maximum weight edge} if it matches the maximum weight edge with probability at least $c$ in star and bi-star instances.
We obtain the following relationship.
Its proof relies on a separate analysis of stars and bi-stars.

\begin{restatable}{lemma}{maxedgeconversion}\label{lem:maxedgeconversion}
	If there exists a $\ctree$-competitive online matching algorithm on the tree domain, 
	then there exists an algorithm for matching the maximum weight edge with a competitive ratio of more than $\frac 13$.
\end{restatable}

Hence, to derive a contradiction to our initial assumption of 
a $\ctree$-competitive online matching algorithm, we prove in the following the nonexistence of algorithms for matching the maximum weight edge with probability more than $\frac 13$.
In the following steps, we want to achieve certain conditions under which our algorithms operate without loss of generality.
This is similar to the reduction by \citet{EFGT22a} to an ``ordinal'' setting.
As a first step, we observe that we can restrict attention to algorithms that, if at all, match the current maximum weight edge in each step.

\begin{lemma}\label{claim:only_max}
	For every star instance, we may assume without loss of generality that only the current maximum weight edge and no negative weight edges are matched.
\end{lemma}

\begin{proof}
	Consider an 
	algorithm $\alg$ for matching the maximum weight edge.
	We modify this algorithm such that whenever it performs a randomized decision to match an edge, it sets probabilities to~$0$ for matching edges that are not currently the maximum weight edge or have negative weight.
	It then continues executing $\alg$ as if the decision of $\alg$ had been performed.
	This algorithm has the desired form, i.e., it only matches the current maximum weight edge and no negative weight edges.
	Moreover, since negative weight edges and edges that are not currently the maximum weight are never the maximum weight edge in star and bi-star instances, the modified algorithm matches the maximum weight edge with the same probability.
\end{proof}

Consequently, we can restrict attention to algorithms that, at each step, face the decision to match the current maximum weight edge, if possible, or do nothing.
From now on, we will only consider such algorithms.

We go one step further and show that when a matching decision is performed (to match a current maximum weight edge), this can be assumed to be independent of how the current state is achieved.

\begin{restatable}{lemma}{HistoryIndependence}\label{claim:history_independent}
	For every star instance, we may assume without loss of generality that our algorithm's decisions only depend on	
	\begin{itemize}
		\item which agents have arrived,
		\item whether $a$ has arrived and is matched, and
		\item whether the last arrived agent is part of the current maximum weight edge.
	\end{itemize}
\end{restatable}

From now on, we consider algorithms as per \Cref{claim:history_independent}.
Finally, we show that algorithmic decisions can be made independently of $b$ and agents associated with $J$.

\begin{restatable}{lemma}{Jindependence}\label{lem:Jindependence}
	For every star instance, we may assume without loss of generality that our algorithms decisions are independent of agents $b$ and agents associated with $J$.
\end{restatable}

From now on, we consider algorithms that, additionally, fulfill the independence of decisions of $b$ and agents associated with $J$. 

The combination of 
\Cref{claim:history_independent,lem:Jindependence} implies that 
an algorithm is fully specified by the matching probabilities dependent on the observed weights but not the arrival orders. 
From now on, we consider a fixed algorithm $\beste$ 
and assume for contradiction that it is $\cmax$-competitive for matching the maximum weight edge with $\cmax > \frac 13$.
It is fully specified by a function $f\colon 2^{\mathbb N}\times \mathbb N\to[0,1]$, where $f$ takes as input a subset $I\subseteq \mathbb N$ (specifying the leaf weights in a star instance) and a positive integer $x$ (specifying the parameter for negative edges).
The value $f(I,x)$ equals the probability of matching the current maximum weight edge provided that $a$ has arrived, is unmatched, the last arrived agent is part of the maximum edge, $a$ has revealed edges precisely to agents corresponding to the set $I$, and $x$ is the parameter for negative edges.

Now, consider a star instance $S \in \starins$ based on parameters $I$, $J$, and $x$ (at this point, $\epsilon$ is irrelevant).
We define 
	\begin{equation*}
		\pmatched(\hrinput) := 
			\prob\left(\{a, d_i\} \in \beste(S)\text{ for some }i\in I\right)\text,
	\end{equation*}
	 
	i.e., the probability to match~$a$.
	The key step is to estimate this quantity.

\begin{restatable}{lemma}{refinedrecursion}\label{lem:refinedrecursion}
	Let $S\in \starins$ with $|I| = k-1$. 
	Then it holds that $\pmatched(\hrinput) > \frac{2}{3} - \frac{2}{3k}$ for all $S \in \starins$.
\end{restatable}

Finally, we want to use the performance on stars to bound the performance on bi-stars.
We essentially use that the prefix of every arrival order in every bi-star is indistinguishable from a star instance until both $a$ and $b$ arrive.

Consider a bi-star instance $B \in \mathcal{B}$ defined by $I$, $J$, and $x$ 
($x$ defines its negative weights), and assume that $|I| = |J|$.
As usual, the number of agents is $n$, i.e., $n = 2 + \lvert I \rvert + \lvert J \rvert$.
Let $Y$ be the random variable that counts the number of agents from $I$ that arrive before~$b$ if $b$ arrives after $a$ and the number of agents from $J$ that arrive before $a$ if $a$ arrives after $b$.
Moreover, let $Y_I$ be the random variable that counts the number of agents from $I$ that arrive before~$b$ and $Y_J$ be the random variable that counts the number of agents from $J$ that arrive before~$a$.

We compute 
\begin{align*}
	&\prob(\{a,b\}\in \beste(B) \mid Y \ge y)\\
	&= \prob(\{a,b\}\in \beste(B) \mid Y \ge y,\sigma^{-1}(a) < \sigma^{-1}(b)) \cdot \prob(\sigma^{-1}(a) < \sigma^{-1}(b)\mid Y \ge y)\\
	&\phantom{=} + \prob(\{a,b\}\in \beste(B) \mid Y \ge y,\sigma^{-1}(b) < \sigma^{-1}(a)) \cdot \prob(\sigma^{-1}(b) < \sigma^{-1}(a)\mid Y \ge y)\\
	& = \frac 12 \cdot \prob(\{a,b\}\in \beste(B) \mid Y_I \ge y,\sigma^{-1}(a) < \sigma^{-1}(b))\\
	&\phantom{=} + \frac 12 \cdot \prob(\{a,b\}\in \beste(B) \mid Y_J \ge y,\sigma^{-1}(b) < \sigma^{-1}(a))\\
	& = \prob(\{a,b\}\in \beste(B) \mid Y_I \ge y,\sigma^{-1}(a) < \sigma^{-1}(b))
\end{align*}
In the last step, we use symmetry between $a$ and $b$ together with $I$ and $J$, which works because $|I| = |J|$.
We thus want to estimate the latter probability.

Note that if $a$ arrives before $b$, then the agents arriving before $b$ form a star instance where the subset of agents of $I$ that has arrived is a uniformly random subset of size $Y_I$.
Hence, by \Cref{lem:refinedrecursion}, we have that 
\begin{align*}
	&\prob(a \text{ matched when }b\text{ arrives}\mid Y_I\ge y,\sigma^{-1}(a) < \sigma^{-1}(b)) > \frac 23 - \frac 2{3 {(y+1)}}\text.
\end{align*}
There, we bound with the worst case where $Y_I = y$, i.e., $k = y + 1$ in \Cref{lem:refinedrecursion}.
It follows that
\begin{small}
\begin{align*}
	&\prob(\{a,b\}\in \beste(B) \mid Y_I \ge y,\sigma^{-1}(a) < \sigma^{-1}(b)) \\
	&\le 1 - \prob(a \text{ matched when }b\text{ arrives}\mid Y_I \ge y,\sigma^{-1}(a) < \sigma^{-1}(b)) < \frac 13 + \frac 2{3 {(y+1)}}\text.
\end{align*}
\end{small}

Clearly, there exists $N\in \mathbb N$ such that for all $y\ge N$, it holds that $\frac{2}{3(y+1)} \le \frac 13 \left(\cmax-\frac 13\right)$.
Together, for all $y\ge N$, we obtain that
\begin{equation}\label{eq:largeN}
	\prob(\{a,b\}\in \beste(B) \mid Y \ge y) < \frac 13 + \frac 13 \left(\cmax-\frac 13\right)\text.
\end{equation}

Second, we want to estimate $\prob(Y < N)$.
Clearly, whenever $Y<N$, then we have that $Y_I < N$ or $Y_J < N$. 
Hence, by a union bound,
\begin{align}
	&\prob(Y < N) \le \prob(Y_I < N\text{ or } Y_J < N) \le \prob(Y_I < N) + \prob(Y_J < N) = 2\prob(Y_I < N)\label{eq:YtoYI}
\end{align}

We now want to bound $\prob(Y_I < N)$. 
Note that $b$ arrives with equal probability in every fixed position among the agents in $I\cup \{b\}$. 
Hence, $Y_I < N$ to happen is equal to $b$ arriving in a position in $\{1,\dots, N\}$ among $I\cup \{b\}$.
We conclude that
	$\prob(Y_I < N) = \frac N{\frac n2} = \frac {2N}n\text.$

Note that this tends to~$0$ for $n$ tending to infinity.
Therefore, there exists $N'\ge N$ such that $\prob(Y_I < N) \le \frac 16 \left(\cmax-\frac 13\right)$ for all $n\ge N'$.
Combining this with \Cref{eq:YtoYI}, for all $n\ge N'$, we obtain
\begin{equation}\label{eq:smallN}
	\prob(Y < N) \le \frac 13 \left(\cmax-\frac 13\right)\text.
\end{equation}
For $n\ge N'$, we conclude that 
\begin{align*}
	&\prob(\{a,b\}\in \beste(B))\\
	& = \prob(\{a,b\}\in \beste(B) \mid Y < N)\prob(Y<N) +  \prob(\{a,b\}\in \beste(B) \mid Y \ge N)\prob(Y \ge N)\\
	&\le \prob(Y<N) + \prob(\{a,b\}\in \beste(B) \mid Y \ge N)
	\overset{\text{Eqs.~(\ref{eq:largeN},\ref{eq:smallN})}}\le \frac 13 \left(\cmax-\frac 13\right) + \left(\frac 13 + \frac 13 \left(\cmax-\frac 13\right)\right)\\
	& \le \frac 13 + \frac 23\left(\cmax-\frac 13\right) = \frac 23\cmax + \frac 19 < \cmax\text.
\end{align*}

This contradicts our assumption that $\beste$ was $\cmax$-competitive.
\end{proof}

\subsection{Lemmas in the Proof of Theorem~\ref{thm:negativeresult_MWMrand}}\label{app:randomarrival:limit}

In this section, we prove auxiliary lemmas in the proof of \Cref{thm:negativeresult_MWMrand}, restated as follows.

\MWMrandNegative*

We start with the proof of \Cref{lem:maxedgeconversion}.
Its proof relies on two auxiliary statements concerning stars and bi-stars.

We first consider stars and want to estimate $\inf_{S \in \starins} \prob(\{a, d_{\starmax}\} \in \gentreealg(S))$, i.e., the infimum of the probability with which the maximum weight edge is matched in stars, where $\gentreealg$ is a $\ctree$-competitive online matching algorithm.

\begin{lemma}\label{lem:cboundstar}
	 Let $\gentreealg$ be a $\ctree$-competitive online matching algorithm.
	 Then it holds that 
	 $$\inf_{S \in \starins} \prob(\{a, d_{\starmax}\} \in \gentreealg(S)) \ge \ctree - \epsilon\text.$$ 
\end{lemma}

\begin{proof}
	Let $\gentreealg$ be a $\ctree$-competitive online matching algorithm.
	Consider some star instance $S \in \starins$.
	By definition of the competitive ratio, it holds that 
	\begin{equation*}
		\frac{\EV[\SW(\gentreealg(S))]}{\SW(\pi^*(S))} = \frac{\EV[\SW(\gentreealg(S))]}{\vb^{\starmax}} \ge \ctree \text,
	\end{equation*}
	where $\pi^*(S)$ denotes the maximum weight matching.
	We compute
		\begin{align*}
			\ctree \cdot\vb^{\starmax} &\le \EV[\SW(\gentreealg(S))] 
			= \sum_{x, y \in N} \prob(\{x, y\} \in \gentreealg(S))\vfd(x,y) \\
			&\le \sum_{i \in I} \prob(\{a, d_i\} \in \gentreealg(S))\vfd(a,d_i) \\
			&= \sum_{i \in I \setminus \{\starmax\}} \prob(\{a, d_i\} \in \gentreealg(S))\vb^i 
			 +  \prob(\{a, d_{\starmax}\} \in \gentreealg(S))\vb^{\starmax}
		\end{align*}
	In the second line, we express the expectation over matchings in terms of single edges.
	The third line follows from the fact that only the valuations between $a$ and the agents associated with $I$ are positive.
	Dividing both sides by $\vb^{\starmax} > 0$, we get
		\begin{align*}
			\ctree &\le  \prob(\{a, d_{\starmax}\} \in \gentreealg(S)) 
			+ \sum_{i \in I \setminus \{\starmax\}} \prob(\{a, d_i\} \in \gentreealg(S))\frac{\vb^i}{\vb^{\starmax}}\\
			&\le \prob(\{a, d_{\starmax}\} \in \gentreealg(S)) 
			+ \sum_{i \in I \setminus \{\starmax\}} \prob(\{a, d_i\} \in \gentreealg(S))\epsilon \\
			&\le \prob(\{a, d_{\starmax}\} \in \gentreealg(S)) + \epsilon \text.
		\end{align*}
	The last inequality follows since $\prob(\{a, x\}  \in \gentreealg(S))$ for $x \in N$ forms a probability distribution since $a$ cannot be matched with probability more than $1$.
	Since $S\in \starins$ was chosen arbitrarily, we obtain $\inf_{S \in \starins} \prob(\{a, d_{\starmax}\} \in \gentreealg(S)) \ge \ctree - \epsilon$.
\end{proof}

Next, we show that $\ctree - 2\epsilon$ is a lower bound on the probability with which 
a $\ctree$-competitive online matching algorithm 
matches the two centers in bi-star instances.
The proof is similar to that of \Cref{lem:cboundstar}.

\begin{lemma}\label{lem:cboundbistar}
	Let $\gentreealg$ be a $\ctree$-competitive online matching algorithm.
	Then it holds that 
	$$\inf_{B \in \biins} \prob(\{a, b\} \in \gentreealg(B)) \ge \ctree - 2\epsilon\text.$$ 
\end{lemma}

\begin{proof}
	Consider a bi-star instance $B \in \biins$.
	Then, by definition of the competitive ratio, it holds that
	\begin{equation*}
		\frac{\EV[\SW(\gentreealg(B))]}{\SW(\pi^*(B))} = \frac{\EV[\SW(\gentreealg(B))]}{\vb^{\bstarmax}} \ge \ctree \text,
	\end{equation*}
	where $\pi^*(B)$ denotes the maximum weight matching.
	We compute
		\begin{align*}
		\ctree \cdot \vb^{\bstarmax + 1} &\le \EV[\SW(\gentreealg(B))] 
		= \sum_{x, y \in N} \prob(\{x, y\} \in \gentreealg(B))\vfd(x,y) \\
			&\le \sum_{i \in I} \prob(\{a, d_i\} \in \gentreealg(B))\vfd(a,d_i) 
			+ \sum_{j \in J} \prob(\{b, d_j\} \in \gentreealg(B))\vfd(b,d_j) \\
			&\phantom{=}+ \prob(\{a, b\} \in \gentreealg(B))\vfd(a,b) \\
			&= \sum_{i \in I} \prob(\{a, d_i\} \in \gentreealg(B)\vb^i 
			+ \sum_{j \in J} \prob(\{b, d_j\} \in \gentreealg(B))\vb^j\\
			&\phantom{=} + \prob(\{a, b\} \in \gentreealg(B)) \vb^{\bstarmax + 1}
		\end{align*}
	In the second line, we express the expectation over matchings in terms of single edges.
	In the subsequent step, we omit edges with large negative weight.
	Dividing both sides by $\vb^{\bstarmax + 1} > 0$, we get
		\begin{align*}
			\ctree &\le  \prob(\{a, b\} \in \gentreealg(B)) 
			+ \sum_{i \in I} \prob(\{a, d_i\} \in \gentreealg(B))\frac{\vb^i}{\vb^{\bstarmax + 1}} 
			+ \sum_{j \in J} \prob(\{b, d_j\} \in \gentreealg(B))\frac{\vb^j}{\vb^{\bstarmax + 1}}\\
			&\le \prob(\{a, b\} \in \gentreealg(B)) 
			+ \sum_{i \in I} \prob(\{a, d_i\} \in \gentreealg(B))\epsilon 
			+ \sum_{j \in J} \prob(\{b, d_j\} \in \gentreealg(B))\epsilon \\
			&\le \prob(\{a, b\} \in \gentreealg(B)) + 2\epsilon\text.
		\end{align*}
	The third inequality follows since $\prob(\{a, x\}  \in \gentreealg(B))$ and $\prob(\{b, x\}  \in \gentreealg(B))$ for $x \in N$ form probability distributions since $a$ and $b$ cannot be matched with probability more than one.
	Since $B \in \biins$ was chosen arbitrarily, we obtain $\inf_{B \in \biins} \prob(\{a, b\} \in \gentreealg(B)) \ge \ctree - 2\epsilon$.
\end{proof}

We can combine \Cref{lem:cboundstar,lem:cboundbistar} to transition to the goal of proving that there is no algorithm matching the maximum weight edge that is better than $\frac 13$-competitive.

\Cref{lem:cboundstar,lem:cboundbistar}, we can transition to the goal to prove that there is no algorithm matching the maximum weight edge that is better than $\frac 13$-competitive.

\maxedgeconversion*
\begin{proof}
	Assume that there is a $\ctree$-competitive online matching algorithm 
	on the tree domain. 
	Consider $c' = \ctree - 2\epsilon$. 
	Since $\epsilon = \frac 13(\ctree - \frac 13)$, it holds that $c' = \frac 13 \ctree - \frac 29 > \frac 13$.
	By \Cref{lem:cboundstar,lem:cboundbistar}, $\alg$ is $c'$-competitive for matching the maximum weight edge.
\end{proof}

Next, we prove our lemma concerning history independence.

\HistoryIndependence*
\begin{proof}
	Consider an algorithm $\alg$ restricted as per \Cref{claim:only_max}.
	We transform this algorithm as follows:
	Consider the arrival of an agent and assume that the algorithm wants to match with positive probability.
	This means that the currently arrived agent is $a$ or the agent of the maximum weight edge.
	Assume that, so far, agents in the set $A$ have arrived.
	Let $H(A)$ be the history of the algorithm so far, which captures the arrival order of agents in $A$ as well as all previous algorithmic decisions.
	Let $\mathcal H(A)$ be the set of all histories where the agents in $A$ arrive such that the last arrived agent is part of the current maximum weight edge, and $a$ is unmatched at the arrival of the last agent.
	
	We obtain a new algorithm $\alg'$ as follows.
	Upon the arrival of an agent that leads to a matching decision in $\alg$ involving agents $A$, the algorithm $\alg'$ ignores the history $H(A)$.
	Instead, it samples a history $H'(A)\sim \mathcal H(A)$ according to the probabilities of this history occurring in $\alg$.\footnote{Note that we are not concerned about the computational complexity for designing this algorithm. Instead, we simply define an algorithm based on the potential randomizations of $\alg$. Note that this technique even applies when $\alg$ is an ``inefficient'' algorithm, i.e., performs computations of any length.}
	Note that this is well-defined as we are operating on a finite game, for which there is only a finite set of histories, and the probabilities of each of the histories occurring only depends on algorithmic (randomized) decisions on all possible histories.
	Then, it matches the current maximum weight edge if and only if $\alg$ would do so given the history $H'(A)$.
	
	By design, we have that $\alg'$ performs for $H(A)$ like $\alg$ performs for $H'(A)$.
	Moreover, the distribution of the sampled histories is identical to the distribution of the real histories.
	Hence, the performance of $\alg'$ in terms of matching the maximum weight edge is identical to the performance of $\alg$.
	However, the decisions of $\alg'$ only depend on the set of agents that has arrived, whether $a$ has arrived and is matched, and whether the last agent is part of the current maximum weight edge. 
\end{proof}

Now, we prove that decisions can be assumed to be independent of $b$ and $J$.

\Jindependence*

\begin{proof}
	Consider an algorithm $\alg$ restricted as per \Cref{claim:only_max}.
	Then, $\alg$ never matches a negative weight edge. 
	Hence, the first matching decision can happen when $a$ arrives, and subsequently, $\alg$ can only match the current maximum weight edge.
	Moreover, once $a$ has arrived, it is revealed which present agents belong to $I$.
	We transform $\alg$ so that every matching decision if it is still possible to match, is made as if $b$ and agents associated with $J$ have not yet arrived.
	In other words, $\alg$ behaves on a star instance with respect to parameters $I$, $J$, $x$, and $\epsilon$, as if $J$ was the empty set.
	Note that the case of the same instance where $J$ really is the empty set is another star instance, and it achieves the same performance as $\alg$ achieved on this instance.
	Hence, its competitive ratio can only improve as it now only depends on a smaller set of star instances.
\end{proof}

Finally, we provide the bound on $\pmatched(\hrinput)$.

\refinedrecursion*

\begin{proof}
	Given a star $S\in \starins$, we additionally define 
	$$\pmatchbest(\hrinput) := \prob(\{a, d_{\starmax}\} \in \beste(S))$$ 
	for all $S \in \starins$, i.e., the probability to match $a$ with $d_{\starmax}$. 
	
	We now show recursive formulas for $\pmatched(\hrinput)$ and $\pmatchbest(\hrinput)$ assuming that we are given a star instance $S\in \starins$ with $|I| = k-1$.
	To this end, we partition all arrival orders in $\Sigma(\{a\}\cup\{d_i\colon i\in I\})$, i.e., of the agents relevant to matching, into three sets based on the last arriving agent.
	The first two are the arrival orders~$\sigma$ 
	in which $a$ or $d_{\starmax}$ arrive last, i.e., $\sigma(\recursepara) = a$ or $\sigma(\recursepara) = d_{\starmax}$, respectively.
	They each make up a $\frac{1}{\recursepara}$ fraction of all arrival orders, i.e., $\prob(\sigma(\recursepara) = a) = \frac{1}{\recursepara}$ and $\prob(\sigma(\recursepara) = d_{\starmax}) = \frac{1}{\recursepara}$.
	In the remaining orders, one of the other alternatives arrives last. 
	We have $\prob(\sigma(\recursepara) \ne a \land \sigma(\recursepara) \ne d_{\starmax}) = \frac{\recursepara - 2}{\recursepara}$.
	Note that for $i \ne \starmax$, if $d_i$ arrives last, then the algorithm cannot match, so it is matched only if it has matched already.
	Furthermore, if $d_{\starmax}$ arrives last, then we need to consider two cases.
	Either $a$ could already be matched or if it is unmatched then we match with probability $f(I,x)$.
	Finally, if $a$ arrives last, then we match with probability $f(I,x)$.

	\begin{align*}
		\pmatched(\hrinput) &= \prob(\{a, d_i\} \in \beste(S)\text{ for some }i\in I)\\
		&= \prob(\{a, d_i\} \in \beste(S)\text{ for some }i\in I \vert \sigma(\recursepara) \neq a \land \sigma(\recursepara) \neq d_{\starmax}) 
		\cdot \prob(\sigma(\recursepara) \neq a \land \sigma(\recursepara) \neq d_{\starmax})\\
		&\phantom{=}+ \prob(\{a, d_i\} \in \beste(S)\text{ for some }i\in I \vert \sigma(\recursepara) = d_{\starmax})\prob(\sigma(\recursepara) = d_{\starmax}) \\ 
		&\phantom{=}+ \prob(\{a, d_i\} \in \beste(S)\text{ for some }i\in I \vert \sigma(\recursepara) = a)\prob(\sigma(\recursepara) = a) \\
		&= \frac{1}{\recursepara}\sum_{i \in I \setminus \{d_{\starmax}\}} \pmatchedi{\recursepara - 1}(\hrinputpara{d_{i}}) 
		+ \frac{1}{\recursepara}\left[\pmatchedi{\recursepara - 1}(\hrinputpara{d_{\starmax}}) + (1 - \pmatchedi{\recursepara - 1}(\hrinputpara{d_{\starmax}}))f(I,x)\right] 
		+ \frac{1}{\recursepara}f(I,x) \\
		&= \frac{1}{\recursepara}\sum_{i \in I} \pmatchedi{\recursepara - 1}(\hrinputpara{d_{i}})  - \frac{f(I,x)}{\recursepara}\pmatchedi{\recursepara - 1}(\hrinputpara{d_{\starmax}}) + \frac{2f(I,x)}{\recursepara}
	\end{align*}
	
	Furthermore, we have $\pmatchedi{2}(S^{x}_{\{d_i\}, J}) =  f(\{a, d_i\},x)$ for all $i \in I$ and the star where the only leaf from $a$ is towards $d_i$ and $J$ is arbitrary.
	We continue by calculating our second term.
	We have
	\begin{align*}
		\pmatchbest(\hrinput) &=  \prob(\{a, d_{\starmax}\} \in \beste(S))\\
		&= \prob(\{a, d_{\starmax}\} \in \beste(S) \vert \sigma(\recursepara) \neq a \land \sigma(\recursepara) \neq d_{\starmax}) 
		\cdot \prob(\sigma(\recursepara) \neq a \land \sigma(\recursepara) \neq d_{\starmax})\\
		&\phantom{=} + \prob(\{a, d_{\starmax}\} \in \beste(S) \vert \sigma(\recursepara) = d_{\starmax})\prob(\sigma(\recursepara) = d_{\starmax}) \\
		&\phantom{=} + \prob(\{a, d_{\starmax}\} \in \beste(S) \vert \sigma(\recursepara) = a)\prob(\sigma(\recursepara) = a) \\
		&= \frac{1}{\recursepara} \sum_{i \in I \setminus \{d_{\starmax}\}} \pmatchbesti{\recursepara - 1}(\hrinputpara{d_{i}}) + \frac{1}{\recursepara}f(I,x)(1 - \pmatchedi{\recursepara - 1}(\hrinputpara{d_{\starmax}})) 
		+ \frac{1}{\recursepara}f(I,x) \\
		&= \frac{1}{\recursepara}\sum_{i \in I \setminus \{d_{\starmax}\}} \pmatchbesti{\recursepara - 1}(\hrinputpara{d_{i}}) - \frac{f(I,x)}{\recursepara}\pmatchedi{\recursepara - 1}(\hrinputpara{d_{\starmax}}) 
		 + \frac{2f(I,x)}{\recursepara} 
	\end{align*}
	In addition, it holds that $\pmatchbesti{2}(S^{x}_{\{d_i\}, J}) = f(\{a, d_i\},x)$ since if $\beste$ matches in this case, then it matches the optimal edge.

	Next, we compute $\pmatched(\hrinput) - \pmatchbest(\hrinput)$, i.e., the probability of matching a suboptimal valuation in a star.
	Some terms will cancel out because the probabilities of matching optimally ($\pmatchbest(\hrinput)$) and matching at all ($\pmatched(\hrinput)$) only differ if the last agent to arrive is not $a$ or $d_{\starmax}$.
	\begin{align*}
		&\pmatched(\hrinput) - \pmatchbest(\hrinput)\\
		 &=  \frac{1}{\recursepara}\sum_{i \in I} \pmatchedi{\recursepara - 1}(\hrinputpara{d_{i}})  - \frac{f(I,x)}{\recursepara}\pmatchedi{\recursepara - 1}(\hrinputpara{d_{\starmax}}) + \frac{2f(I,x)}{\recursepara} \\
		& - \frac{1}{\recursepara}\sum_{i \in I \setminus \{d_{\starmax}\}} \pmatchbesti{\recursepara - 1}(\hrinputpara{d_{i}}) + \frac{f(I,x)}{\recursepara}\pmatchedi{\recursepara - 1}(\hrinputpara{d_{\starmax}}) 
		 - \frac{2f(I,x)}{\recursepara}  \\
		&= \frac{1}{\recursepara}\sum_{i \in I} \pmatchedi{\recursepara - 1}(\hrinputpara{d_{i}}) - \frac{1}{\recursepara}\sum_{i \in I \setminus \{d_{\starmax}\}} \pmatchbesti{\recursepara - 1}(\hrinputpara{d_{i}}) \\
		&= \frac{1}{\recursepara}\pmatchedi{\recursepara - 1}(\hrinputpara{d_{\starmax}}) + \frac{1}{\recursepara}\sum_{i \in I \setminus \{d_{\starmax}\}} \pmatchedi{\recursepara - 1}(\hrinputpara{d_{i}}) - \pmatchbesti{\recursepara - 1}(\hrinputpara{d_{i}})
	\end{align*}
	
	We can repeatedly apply the recursive equation that we just derived.
	On the right side, this amounts to summing 
	$$\frac{1}{\recursepara}\frac{(\recursepara - 1 - \lvert \tilde{I} \rvert)!\lvert \tilde{I} \rvert !}{(\recursepara - 1)!}\pmatchedi{\recursepara - 1}(S[N \setminus (\{d_{\starmax} \cup \tilde{I} )\}]) $$
	for all $\tilde{I} \subsetneq I \setminus \{d_{\starmax}\}$.
	The factor $\frac{(\recursepara - 1 - \lvert \tilde{I} \rvert)!}{(\recursepara - 1)!}$ collects the accumulated prefactors of all steps, and the factor $\lvert \tilde{I} \rvert !$ accounts for the fact that we can arrive at the same term by removing the alternatives in $\tilde{I}$ in any order.
	Finally, the remaining difference after removing all elements in $I \setminus \{d_{\starmax}\}$ is  $\pmatchedi{2}(S^{x}_{\{d_{\starmax}\}, \tilde{I}}) - \pmatchbesti{2}(S^{x}_{\{d_{\starmax}\}, \tilde{I}}) = 0$, which cancels out.
	We can rewrite $\frac{(\recursepara - 1 - \lvert \tilde{I} \rvert)!\lvert \tilde{I} \rvert !}{(\recursepara - 1)!} = \frac{1}{\binom{k - 1}{\lvert \tilde{I} \rvert}}$.
	This yields:
	\begin{equation}\label{eq:hrdiff}
		\pmatched(\hrinput) - \pmatchbest(\hrinput) = \frac{1}{\recursepara} \sum_{\tilde{I} \subsetneq I \setminus \{d_{\starmax}\}}\frac{1}{\binom{k - 1}{\lvert \tilde{I} \rvert}}\pmatchedi{\recursepara - 1}(S[N \setminus (\{d_{\starmax} \cup \tilde{I} )\}])
	\end{equation}
	
	Define $||S|| := |I| + 1$ if $S$ is a star defined by $I$, i.e., $||S|| = |I\cup \{a\}|$.
	Hence, we have that $||S|| = k$.

	We now show the lemma by strong induction over $||S||$.
	Note that $I\neq \emptyset$ and, therefore, $||S||\ge 2$ in all star instances.
	If $||S|| = 2$, then $\pmatchedi{2}(S) = \pmatchbesti{2}(S)$.
	Thus, $$\pmatchedi{2}(S) = \pmatchbesti{2}(S) \ge \cmax > \frac{1}{3} = \frac{2}{3} - \frac{2}{3 \cdot 2}\text.$$ 
	
	Now assume for all stars $S$ with $||S||\le k-1$, it holds that $\pmatched(S) > \frac{2}{3} - \frac{2}{3||S||}$.
	In the following, we use the binomial identity 
	\begin{equation}\label{eq:binid}
		\binom{n - 1}{k} = \frac{n - k}{n}\binom{n}{k}\text.
	\end{equation}
	
	Recall that $|I| = k-1$ and, therefore, $|I\setminus \{d_{\starmax}\}| = k-2$.
	We compute
	\begin{align*}
		\pmatched(\hrinput) &\overset{Eq.~(\ref{eq:hrdiff})}{=} \pmatchbest(\hrinput) + \frac{1}{\recursepara} \sum_{\tilde{I} \subsetneq I \setminus \{d_{\starmax}\}} \frac{1}{\binom{k - 1}{\lvert \tilde{I} \rvert}}\pmatchedi{\recursepara - 1}(S[N \setminus (\{d_{\starmax} \cup \tilde{I} )\}])\\
		&> \frac{1}{3} + \frac{1}{\recursepara} \sum_{\tilde{I} \subsetneq I \setminus \{d_{\starmax}\}} \frac{1}{\binom{k - 1}{\lvert \tilde{I} \rvert}}\left(\frac{2}{3} - \frac{2}{3(k - 1 - \lvert \tilde{I} \rvert)}\right)\\
		&= \frac{1}{3} + \frac{1}{\recursepara} \sum_{i = 0}^{k - 3} \frac{\binom{k - 2}{i}}{\binom{k - 1}{i}}\left(\frac{2}{3} - \frac{2}{3(k - 1 - i)}\right)\\
		&\overset{Eq.~(\ref{eq:binid})}{=} \frac{1}{3} + \frac{1}{\recursepara} \sum_{i = 0}^{k - 3} \frac{\binom{k - 1}{i}}{\binom{k - 1}{i}}\frac{k - 1 - i}{k - 1}\left(\frac{2}{3} - \frac{2}{3(k - 1 - i)}\right)\\
		&=\frac{1}{3} + \frac{1}{\recursepara} \sum_{i = 0}^{k - 3} \left(\frac{k - 1}{k - 1}  - \frac{i}{k - 1} \right)\left(\frac{2}{3} - \frac{2}{3(k - 1 - i)}\right)\\
		&=\frac{1}{3} + \frac{1}{\recursepara} \sum_{i = 0}^{k - 3} \left(1  - \frac{i}{k - 1} \right)\left(\frac{2}{3} - \frac{2}{3(k - 1 - i)}\right)\\
		&=\frac{1}{3} + \frac{1}{\recursepara} \sum_{i = 0}^{k - 2} \left(1  - \frac{i}{k - 1} \right)\left(\frac{2}{3} - \frac{2}{3(k - 1 - i)}\right)\text.
	\end{align*}
	
	In the last step, we inserted the term for $i = k-2$, which evaluates to~$0$ as $\frac{2}{3(k - 1 - (k-2))} = \frac 23$.
	
	We finally simplify the two parts of the equation individually.
	For the first term, we obtain
	
	\begin{align*}
		\frac{1}{\recursepara} \sum_{i = 0}^{k - 2} \left(1  - \frac{i}{k - 1} \right)\frac{2}{3} &= \frac{2}{3\recursepara}\left( \sum_{i = 0}^{k - 2} 1  - \sum_{i = 0}^{k - 2}\frac{i}{k - 1}\right)
		=\frac{2}{3\recursepara}\left(k - 1 - \frac{1}{k - 1}\frac{(k - 1)(k - 2)}{2} \right) \\
		&= \frac{2}{3\recursepara}\frac{2k - 2 - k + 2}{2} 
		= \frac{2}{3\recursepara}\frac{k}{2} = \frac{1}{3}
	\end{align*}
	
	For the second term, we obtain
	\begin{align*}
		\frac{1}{\recursepara} \sum_{i = 0}^{k - 2} \left(1  - \frac{i}{k - 1} \right)\frac{2}{3(k - 1 - i)} &= \frac{2}{3\recursepara} \sum_{i = 0}^{k - 2} \frac{1 - \frac{i}{k - 1}}{k - 1 - i}
		=\frac{2}{3\recursepara} \sum_{i = 0}^{k - 2} \frac{\frac{k - 1 - i}{k - 1}}{k - 1 - i} \\
		&= \frac{2}{3\recursepara} \sum_{i = 0}^{k - 2} \frac{1}{k - 1} 
		= \frac{2}{3\recursepara}
	\end{align*}
	
	Inserting back into our equation we get
	
	\begin{equation*}
		\pmatched(\hrinput) > \frac{1}{3} + \frac{1}{\recursepara} \sum_{i = 0}^{k - 2} \left(1  - \frac{i}{k - 1} \right)\left(\frac{2}{3} - \frac{2}{3(k - 1 - i)}\right) = \frac{1}{3} + \frac{1}{3} - \frac{2}{3\recursepara} = \frac{2}{3} - \frac{2}{3\recursepara}\text.
	\end{equation*}
	This completes the proof.
\end{proof}

\section{Full proof of Theorem~\ref{thm:dissolutionbound}}\label{app:dissolutionbound} 

Our proof of \Cref{thm:dissolutionbound} relies on a similar idea as the proof by \citet{Vara11a}, showing that there does not exist an online matching algorithm (in an edge arrival setting) operating under free dissolution for which the competitive ratio is better than $\frac 1{3+ 2 \sqrt{2}}$.
His proof relies on two steps.
First, he shows that a particular sequence of real numbers cannot exist based on a recursive set of inequalities.
Second, he shows that the existence of an algorithm with a competitive ratio of better than  $\frac 1{3+ 2 \sqrt{2}}$ implies the existence of just such a sequence.
We will use his first step as a black box and then use an adversarial instance of online FHGs to construct the sequence utilizing an online coalition formation algorithm that achieves a competitive ratio of better than $\frac 1{6+ 4 \sqrt{2}}$.
The construction of our adversarial instance is similar to the one by \citet{Vara11a}.
Still, while his optimal partition is a matching consisting of coalitions of size~$2$, we construct the instance in a way such that the optimal instance consists of coalitions that form stars (i.e., we have symmetric valuations that are equal to some constant if they involve a special center agent and are~$0$, otherwise).
This accounts for the improvement of about a factor of~$2$ in the welfare of the optimal partition.

We start by stating the lemma that captures the nonexistence of the sequence.

\begin{lemma}[\citet{Vara11a}]\label{lem:sequence}
	Let $\beta > \frac 1{3 + 2 \sqrt{2}}$. Then there exists no sequence $(x_i)_{i\in \mathbb N}$ with $x_1 = 1$ and $x_i \ge 0$ for $i\ge 2$ such that for all $i\in \mathbb N$, it holds that
	\begin{equation}
		x_i \ge \beta\left(x_{i+1} + \sum_{j=1}^{i+1}x_j\right)\text.
	\end{equation}
\end{lemma}

Next, we evaluate the social welfare of a ``star'' coalition.

\begin{lemma}\label{lem:stareval}
	Let $x\in \mathbb R$. 
	Consider a set of agents $C$ such that there exists $a\in C$ with symmetric valuations $\vf(a,b) = x$ for all $b\in C\setminus \{a\}$ and $\vf(b,b') = 0$ for all $b,b'\in C\setminus \{a\}$ with $b \neq b'$.
	Then it holds that $\SW(C) = 2\frac{|C|-1}{|C|}x$.
\end{lemma}

\begin{proof}
	Assume that we are in the lemma's situation.
	Then, $\uf_a(C) = \frac{|C|-1}{|C|}x$, and for all $b\in C\setminus\{a\}$, it holds that $\uf_b(C) = \frac 1{|C|}x$.
	The assertion follows by summing up utilities.
\end{proof}

We are ready to prove our theorem.

\dissolutionbound*
	
	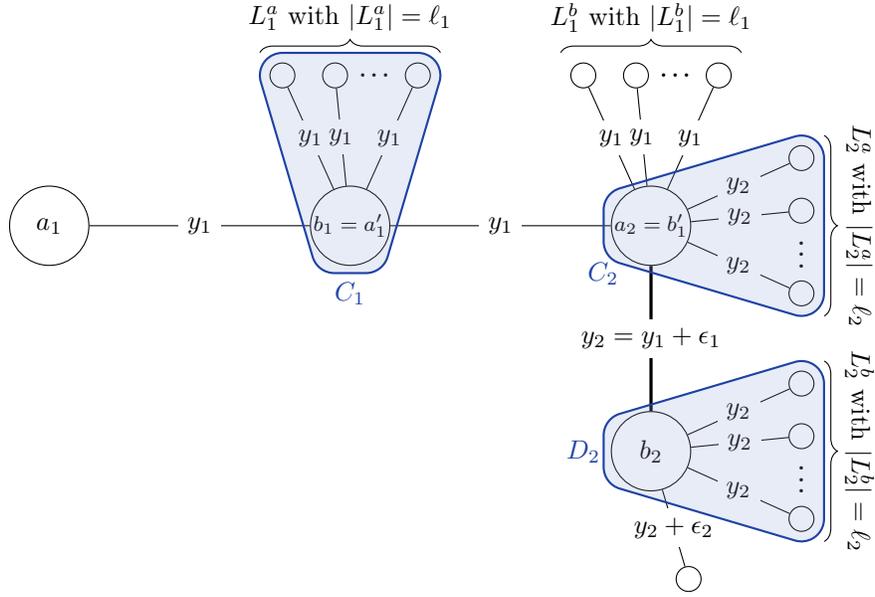
\begin{figure*}
		\centering
		\begin{tikzpicture}
			\node[agent] (a1) at (0,0) {$a_1$};
			\node[agent] (b1) at (4,0) {};
			\node at (b1) {\footnotesize $b_1 = a'_1$};
			\node[agent] (a2) at (8,0) {};
			\node at (a2) {\footnotesize $a_2 = b'_1$};
			\node[agent] (b2) at (8,-3) {$b_2$};
			
			\node at ($(b1) + (0,-.9)$) {\color{CoalitionColor}{$C_1$}};
			\node at ($(a2) + (225:.9)$) {\color{CoalitionColor}{$C_2$}};
			\node at ($(b2) + (180:.9)$) {\color{CoalitionColor}{$D_2$}};
			
			\draw (a1) edge node[midway,fill = white] {$y_1$} (b1);
			\draw (a2) edge node[midway,fill = white] {$y_1$} (b1);
			\draw (a2) edge[very thick] node[midway,fill = white] {$y_2 = y_1 + \epsilon_1$} (b2);
			
			\node[aux] (l1) at ($(b1) + (-.9,2)$) {};
			\node[aux] (l2) at ($(b1) + (-.2,2)$) {};
			\node[aux] (l3) at ($(b1) + (.9,2)$) {};
			\node at ($(l2)!.5!(l3)$) {\footnotesize \textbf{\dots}};
			\draw [decorate,decoration={brace,amplitude=5pt,raise=2ex}]
			  ($(l1)-(.3,0)$) -- ($(l3) + (.3,0)$) node[midway,yshift=2.2em]{$L_1^a$ with $|L_1^a| = \ell_1$};
			  
			\foreach \i in {1,2,3}
			{\draw (b1) edge node[midway,fill = white] {$y_1$} (l\i);}
			
			\node (b12) at ($(b1)-(-.21,.33)$) {};
			\node (b11) at ($(b1)-(.21,.33)$) {};
			\draw[thick,CoalitionColor, fill=CoalitionColor!50, fill opacity=0.2]  \convexpath{l1,l3,b12,b11}{.3cm};
			
			\node[aux] (l1) at ($(a2) + (-.9,2)$) {};
			\node[aux] (l2) at ($(a2) + (-.2,2)$) {};
			\node[aux] (l3) at ($(a2) + (.9,2)$) {};
			\node at ($(l2)!.5!(l3)$) {\footnotesize \textbf{\dots}};
			\draw [decorate,decoration={brace,amplitude=5pt,raise=2ex}]
			  ($(l1)-(.3,0)$) -- ($(l3) + (.3,0)$) node[midway,yshift=2.2em]{$L_1^b$ with $|L_1^b| = \ell_1$};
			
			\foreach \i in {1,2,3}
			{\draw (a2) edge node[midway,fill = white] {$y_1$} (l\i);}
			
			\node[aux] (l1) at ($(a2) + (2,.9)$) {};
			\node[aux] (l2) at ($(a2) + (2,.2)$) {};
			\node[aux] (l3) at ($(a2) + (2,-.9)$) {};
			\node[rotate = 90] at ($(l2)!.5!(l3)$) {\footnotesize \textbf{\dots}};
			\draw [decorate,decoration={brace,amplitude=5pt,raise=2ex}]
			  ($(l1)+(0,.3)$) -- ($(l3)-(0,.3)$) node[midway,xshift=2.2em, rotate = 270]{$L_2^a$ with $|L_2^a| = \ell_2$};
			
			\foreach \i in {1,2,3}
			{\draw (a2) edge node[midway,fill = white] {$y_2$} (l\i);}
			
			\node (a22) at ($(a2)-(.33,-.21)$) {};
			\node (a21) at ($(a2)-(.33,.21)$) {};
			\draw[thick,CoalitionColor, fill=CoalitionColor!50, fill opacity=0.2]  \convexpath{l1,l3,a21,a22}{.3cm};
			
			\node[aux] (l1) at ($(b2) + (2,.9)$) {};
			\node[aux] (l2) at ($(b2) + (2,.2)$) {};
			\node[aux] (l3) at ($(b2) + (2,-.9)$) {};
			\node[rotate = 90] at ($(l2)!.5!(l3)$) {\footnotesize \textbf{\dots}};
			\draw [decorate,decoration={brace,amplitude=5pt,raise=2ex}]
			  ($(l1)+(0,.3)$) -- ($(l3)-(0,.3)$) node[midway,xshift=2.2em, rotate = 270]{$L_2^b$ with $|L_2^b| = \ell_2$};
			
			\node[aux] (l4) at ($(b2) + (.5,-1.7)$) {};
			\draw (b2) edge node[midway,fill = white] {$y_2 + \epsilon_2$} (l4);
			
			\foreach \i in {1,2,3}
			{\draw (b2) edge node[midway,fill = white] {$y_2$} (l\i);}
			
			\node (b22) at ($(b2)-(.33,-.21)$) {};
			\node (b21) at ($(b2)-(.33,.21)$) {};
			\draw[thick,CoalitionColor, fill=CoalitionColor!50, fill opacity=0.2]  \convexpath{l1,l3,b21,b22}{.3cm};

		\end{tikzpicture}
		\caption{Illustration of the construction in the proof of \Cref{thm:dissolutionbound} for an exemplary algorithm $\alg$.
		We display all positive valuations.
		The remaining valuations within the leaf sets $L_1^a$, $L_1^b$, $L_2^a$, and $L_2^b$ are zero, and all other valuations are large negative numbers.
		We start with two agents, $a_1$ and $b_1$.
		We first attempt to dispatch a set $L_1^b$ of leaves towards $b_1$.
		However, our algorithm might immediately decide to dissolve $\{a_1,b_1\}$ and create a new coalition $\{a'_1,b'_1\}$.
		We then might be able to have all the leaf agents in $L_1^a$ and $L_1^b$ arrive.
		This completes the first part of Phase~$1$.
		Now, we start the second part, in which we subsequently increment the valuations. 
		$\alg$ might decide to immediately dissolve $\{a'_1,b'_1\}$ when the next agent arrives.
		This defines agents $a_2$, $b_2$, and coalition $C_1$.
		We start with Phase~$2$.
		In the first part, the leaf agents $L_2^a$ and $L_2^b$ might arrive without further interruption.
		Now assume that $\alg$ would dissolve $\{a_2,b_2\}$ when the next agent arrives (their edge is indicated in bold).
		This would give rise to the definition of $C_2$ and $D_2$, and we would obtain an inequality for $y_2$ by comparing with the guarantee for the coalition structure containing the nonempty coalitions $C_1$, $C_2$, and $D_2$.
		\label{fig:dissolutionexample}
		} 
	\end{figure*}
	
\begin{proof}
	Let $c := \frac 1{6 + 4 \sqrt{2}}$. Assume for contradiction that $\alg$ is an online coalition formation algorithm operating under free dissolution that achieves a competitive ratio of $\gamma > c$ for symmetric FHGs.
	Without loss of generality, we may assume that $\frac c{\gamma}$ is rational.\footnote{Indeed, otherwise, we can just perform the proof for a $\gamma'$ in the open interval $(c,\gamma)$ with this property. Such a $\gamma'$ exists as the function $f:[c,\gamma] \to \mathbb R$, $f(x) = \frac cx$ is continuous and hence, by the density of the rational numbers in the real numbers, attains rational numbers in the open interval $(c,\gamma)$.
	}
	We want this property to assure that all instances we construct exclusively use rational numbers.

	Let 
	\begin{equation}\label{eq:betadef}
		\beta := 2 \left(c + \frac 12\left(\gamma - c\right)\right) = \gamma + c\text,
	\end{equation}
	i.e., it holds that $\beta > 2c = \frac 1{3+ 2 \sqrt{2}}$.
	We will eventually derive a contraction to \Cref{lem:sequence} by constructing a sequence for this~$\beta$.
	
	We construct an adversarial instance for this algorithm by constructing a symmetric graph $G = (N,\vf)$, i.e., we specify the symmetric weights underlying the valuations of an FHG.

	The construction maintains the property that the algorithm's current partition can only contain a single coalition with positive welfare and that coalition contains exactly two agents.
	The adversarial instance is constructed in a sequence of phases, where in every phase, we grow star-like structures around each of the endpoints of the currently maintained nonsingleton coalition.
	In the first part of Phase~$i$, we achieve a star with $\ell_i$ leaves, while the algorithm does not change the matched edges.
	In the second part of Phase~$i$, we iteratively increase the weight on the edges of the stars by $\epsilon_i$ until the algorithm changes the matched edge.
	This has to happen eventually because the algorithm achieves a bounded competitive ratio.
	
	We now specify the two parameters of the construction.
	For $i\in \mathbb N$, define 
	\begin{equation}\label{eq:epsdef}
		\epsilon_i := \frac{\gamma -c}{2\gamma}2^{-i}\quad\text{and}\quad \ell_i := \left\lceil \frac {1-\epsilon_i}{\epsilon_i}\right\rceil\text.
	\end{equation}
	
	Note that, since $\frac c{\gamma}$ is rational, $\epsilon_i = \left(\frac 12 - \frac c{\gamma}\right)2^{-i}$ is also rational.
	Moreover, the definition of $\ell_i$ immediately implies that
	\begin{equation}\label{eq:lbound}
		\frac{\ell_i}{\ell_i + 1}\ge 1-\epsilon_i\text.
	\end{equation}
	
	We now specify the instance.
	Our whole construction is illustrated in \Cref{fig:dissolutionexample}.
	
	The first two agents that arrive are $a_1$ and $b_1$ such that $\vf(a_1,b_1) = 1$.
	Clearly, $\alg$ has to form the coalition $\{a_1,b_1\}$ as otherwise, its competitive ratio would be unbounded.
	For $i\ge 1$, at the beginning of Phase~$i$, there is a single coalition with nonzero welfare containing precisely agents $a_i$ and $b_i$.
	
	Moreover, throughout the execution of the instance, all arriving agents will have a positive (mutual) valuation for precisely one agent---one of the agents that presently is in a coalition of positive welfare---, a zero valuation for some agents, and a large negative valuation for all other agents.
	In particular, the second agent in the coalition of positive welfare yields a large negative valuation, and thus, joining this coalition leads to an overall negative welfare, which cannot be performed by any algorithm with a positive competitive ratio.
	Hence, the new agent only forms a coalition of positive welfare if the previously existing coalition with positive welfare is dissolved.
	
	Now let $i\ge 1$ and assume that we are at the beginning of Phase~$i$, i.e., so far $\alg$ has constructed a partition containing a single coalition with positive welfare containing $a_i$ and $b_i$.
	We set 
	\begin{equation}
		y_i := \vf(a_i,b_i)\text.
	\end{equation}
	In the first part of Phase~1, we want to guarantee that at the end of this part, there is a single coalition of positive welfare
	$C = \{a'_i,b'_i\}$ such that for each of $a'_i$ and $b'_i$, $\ell_i$ agents have arrived such that there are $0$-valuations among these agents and a valuation of $y_i$ towards $a'_i$ or $b'_i$.
	In other words, the instance contains a bi-star as a substructure where all edges weigh $y_i$.
	
	We start by setting $a'_i := a_i$ and $b'_i := b_i$. 
	Now, we let arrive a set $L_i^b$ of up to $\ell_i$ agents that have a valuation of $y_i$ for $b_i$, $0$ for already arrived agents in $L_i^b$, and a sufficiently large negative valuation for all other agents, e.g., a negative value larger in absolute value than the sum of positive valuations of already existing agents.
	As we argued before, the only way that $\alg$ puts an agent in $L_i^b$ into a coalition of positive welfare is if the coalition of $a'_i$ and $b'_i$ is dissolved and the new agent forms a coalition with $b'_i$.
	In this case, we update agent labels: $b'_i$ becomes the new $a'_i$, and the newly arrived agent is the new $b'_i$.
	
	We repeat this until $\ell_i$ agents have arrived.
	Note that this has to happen at some point as we would otherwise have a path of unbounded length with edge weights equal to $y_i$, which would give rise to a partition of social welfare more than $\frac 1{\gamma} y_i$, a contradiction.
	
	Now, we repeat the same procedure with $a'_i$: we let arrive a set $L_i^a$ of up to $\ell_i$ agents that have a valuation of $y_i$ for $a_i$, $0$ for already arrived agents in $L_i^a$, and a sufficiently large negative valuation for all other agents.
	If the algorithm decides to dissolve $\{a'_i,b'_i\}$ to form a coalition of $a'_i$ with a newly arrived agent, we update agent labels: $a'_i$ stays the new $a'_i$, and the newly arrived agent is the new $b'_i$.
	Note that this part must eventually end with all $\ell_i$ agents having arrived.
	Otherwise, we have an unbounded number of agents that at some point had the role of $b'_i$, and each of them can form a coalition with an agent in their set $L_i^b$, which yields unbounded welfare.
	
	We reach the end of the first part of Phase~$i$ and have established a pair of agents $\{a'_i,b'_i\}$ together with their sets $L_i^a$ and $L_i^b$.
	Note that the coalitions $\{a'_i\}\cup L_i^a$ and $\{b'_i\}\cup L_i^b$ are ``star'' coalitions as in the prerequisites of \Cref{lem:stareval}.
	
	We now start the second part of Phase~$i$.
	In this part, further agents arrive that are new leaves to $a'_i$ and $b'_i$.
	Compared to $y_i$, the weight on their connecting edges is increased by $\epsilon_i$.
	We continue until for each of  $a'_i$ and $b'_i$, further sets of $\ell_i$ agents have arrived.
	This means that we have now grown starts with a slightly larger value.
	We repeat the same with increasingly larger valuations in further increments of $\epsilon_i$. 
	The second part of Phase~$i$ ends once throughout this procedure, the edge $\{a'_i,b'_i\}$ gets dissolved. 
	
	We now formalize this idea.
	Set $L_i^{a,0} := L_i^a$ and $L_i^{b,0} := L_i^b$.
	We proceed as follows until the algorithm dissolves a coalition and forms a new coalition of positive welfare.
	For each $j\ge 1$, once all agents in the sets $L_i^{a,j-1}$ and $L_i^{b,j-1}$ have arrived, we proceed as follows.
	We let a set $L_i^{a,j}$ with $\ell_i$ agents arrive that have a valuation of $y_i + j\epsilon_i$ for $a_i$, $0$ for already arrived agents in $L_i^{a,j}$, and a sufficiently large negative valuation for all other agents.
	These agents arrive one by one, so the phase can end before all agents in $L_i^{a,j}$ have arrived.
	Then we let a set $L_i^{b,j}$ with $\ell_i$ agents arrive that have a valuation of $y_i + j\epsilon_i$ for $b_i$, $0$ for already arrived agents in $L_i^{b,j}$, and a sufficiently large negative valuation for all other agents.
	
	Note that this part also has to terminate at some point as otherwise agents with an unbounded valuation arrive, leading to a partition of welfare higher than $\frac 1{\gamma} y_i$.
	
	Once the algorithm forms a new coalition---say this happens when the $j^*$th sets of agents arrive---we distinguish two cases:
	If $a'_i$ remains in a nonsingleton coalition with the new agent $z$, we define
	$C_i := \{b'_i\}\cup L_i^{b,j^*-1}$ and $D_i := \{a'_i\}\cup L_i^{a,j^*-1}$ and set $a_{i+1} = a'_i$ and $b_{i+1} = z$.
	Otherwise, if $b'_i$ remains in a nonsingleton coalition with the new agent $z$, we define
	$C_i := \{a'_i\}\cup L_i^{a,j^*-1}$ and $D_i := \{b'_i\}\cup L_i^{b,j^*-1}$ and set $a_{i+1} = b'_i$ and $b_{i+1} = z$.
	
	Then, the new agents $a_{i+1}$ and $b_{i+1}$ are the only agents in a coalition of positive welfare $y_{i+1} = \vf(a_{i+1},b_{i+1})$.
	Moreover, $C_i$ and $D_i$ are ``star'' coalitions that are disjoint from all previous coalitions $C_k$ for $k < i$ and where all nonzero valuations are $y_{i+1} - \epsilon_i$.
	By \Cref{lem:stareval}, we obtain
	\begin{equation}\label{eq:stareval}
		\SW(C_i) = \SW(D_i) = 2\frac{\ell_i}{\ell_i+1}(y_{i+1}-\epsilon_i)\text.
	\end{equation}
	
	Consider the partition $\pi_i$ containing the coalitions $D_i$, $C_j$ for $1\le j\le i$, and singleton coalitions for all agents not contained in these.
	This coalition already exists right before the arrival of the agent such that the coalition $\{a'_i,b'_i\}$ is dissolved.
	Note that at this point, the social welfare of the partition created by $\alg$ is $y_i$, where we add $\frac {y_i}2$ for each of $a'_i$ and $b'_i$.
	Since $\alg$ is $\gamma$-competitive, we obtain
	\begin{align*}
		& y_i \ge \gamma \cdot \SW(\pi_i) = \gamma \left(\SW(D_{i}) + \sum_{j = 1}^i \SW(C_j)\right)\\
		& \overset{\text{(\ref{eq:stareval})}}{=} \gamma \left(2\frac{\ell_i}{\ell_i+1}(y_{i+1}-\epsilon_i) + \sum_{j = 1}^i 2\frac{\ell_j}{\ell_j+1}(y_{j+1}-\epsilon_j)\right)\\
		& \overset{\text{(\ref{eq:lbound})}}{\ge} \gamma \left(2(1-\epsilon_i)(y_{i+1}-\epsilon_i) + \sum_{j = 1}^i 2(1-\epsilon_j)(y_{j+1}-\epsilon_j)\right)\\
		& \ge \gamma \left(2(y_{i+1}-2y_{i+1}\epsilon_i) + \sum_{j = 1}^i 2(y_{j+1}-2y_{j+1}\epsilon_j)\right)\\
		& \ge \gamma \left(2(y_{i+1}-2y_{i+1}\epsilon_i) + \sum_{j = 1}^i 2(y_{j+1}-2y_{i+1}\epsilon_j)\right)\\
		& = 2\gamma \left(y_{i+1} + \sum_{j = 1}^i y_{j+1}\right)-2\gamma y_{i+1}\left(\epsilon_i + \sum_{j = 1}^i \epsilon_j\right)\\
		& \overset{\text{(\ref{eq:betadef}), (\ref{eq:epsdef})}}{=} (\beta + (\gamma - c)) \left(y_{i+1} + \sum_{j = 1}^i y_{j+1}\right)
		 -2\gamma y_{i+1}\left( \frac{\gamma -c}{2\gamma}2^{-i}+ \sum_{j = 1}^i \frac{\gamma -c}{2\gamma}2^{-j}\right)\\
		& \ge \beta \left(y_{i+1} + \sum_{j = 1}^i y_{j+1}\right)
		 + (\gamma -c)y_{i+1} - (\gamma - c)y_{i+1}\left(2^{-i} + \sum_{j = 1}^i 2^{-j}\right)\\
		& = \beta \left(y_{i+1} + \sum_{j = 1}^i y_{j+1}\right)\text.
	\end{align*}
	
	We obtain our desired sequence by scaling the $y_i$ and starting with $y_2$. 
	Formally,
	for $i\in \mathbb N$, we set $x_i := \frac{y_{i+1}}{y_2}$.
	Then, $x_1 = \frac {y_2}{y_2} = 1$ and for $i\ge 2$, it holds that $x_i \ge 0$.
	Moreover, for $i\ge 1$, our previous calculation implies that
	\begin{align*}
		x_i &= \frac{y_{i+1}}{y_2} \ge \frac 1{y_2}\beta \left(y_{i+2} + \sum_{j=2}^{i+2}y_j\right)= \beta\left(\frac{y_{i+2}}{y_2} + \sum_{j=2}^{i+2}\frac{y_j}{y_2}\right) \\
		& = \beta\left(x_{i+1} + \sum_{j=2}^{i+2}x_{j-1}\right)
		 = \beta\left(x_{i+1} + \sum_{j=1}^{i+1}x_{j}\right) \text.
	\end{align*} 
	Hence, we have constructed the desired sequence. 
	This is a contradiction to \Cref{lem:sequence}.
\end{proof}

\section*{Acknowledgements}
	Most of this work was done when Martin Bullinger was at the University of Oxford.
	Martin Bullinger was supported by the AI Programme of The Alan Turing Institute. 
	Ren\'{e} Romen and Alexander Schlenga were supported by the Deutsche Forschungsgemeinschaft under grant BR~2312/11--2.

\end{document}